\pdfoutput=1 
\documentclass[11pt,draftcls,onecolumn]{IEEEtran}
\usepackage{amssymb, amsmath, graphicx, cite, color, epsfig, subfigure,wrapfig}

\newtheorem{theorem}{Theorem}
\newtheorem{lemma}[theorem]{Lemma}

\newtheorem{corollary}[theorem]{Corollary}

\newtheorem{remark}{Remark}

\allowdisplaybreaks[1]

\begin{document}
\title{Lattice codes for the Gaussian relay channel: Decode-and-Forward and Compress-and-Forward}
\author{Yiwei Song and Natasha Devroye%
\thanks{Yiwei Song and
Natasha Devroye are with the Department of Electrical and Computer
Engineering, University of Illinois at Chicago, Chicago, IL 60607.
Email:~ysong34, devroye@uic.edu.
The work of N. Devroye and Y. Song was partially supported by NSF under awards CCF-1216825 and 1053933. The contents of this article are solely the responsibility of the authors and do not necessarily represent the official views of the NSF. This paper was presented in part at \cite{Song:2010:latticerelay,Song:CF}. 
}}

\maketitle

\begin{abstract}
Lattice codes are known to achieve capacity in the Gaussian point-to-point channel, achieving the same rates as independent, identically distributed (i.i.d.) random Gaussian codebooks. Lattice codes are also known to outperform random codes for certain channel models that are able to exploit their linearity. In this work, we show that lattice codes may be used to achieve the same performance as known i.i.d. Gaussian random coding techniques for the Gaussian relay channel, and show several examples of how this may be combined with the linearity of lattices codes in multi-source relay networks. 
In particular, we present a nested lattice list decoding technique, by which, lattice codes are shown to achieve the Decode-and-Forward (DF) rate of single source, single destination Gaussian relay channels with one or more relays. 
We next present two examples of how this DF scheme may be combined with the linearity of lattice codes to achieve new rate regions which for some channel conditions outperform analogous known Gaussian random coding techniques in multi-source relay channels. That is,  we derive a new achievable rate region for the two-way relay channel with direct links and compare it to existing schemes,  and derive another achievable rate region for the multiple access relay channel.
We furthermore present a lattice Compress-and-Forward (CF) scheme for the Gaussian relay channel which exploits a lattice Wyner-Ziv binning scheme and  achieves the same rate
 as the Cover-El Gamal CF rate evaluated for Gaussian random codes. 
 These results suggest that structured/lattice codes may be used to mimic, and sometimes outperform, random Gaussian codes in 
 general Gaussian networks.  
\end{abstract}

\begin{keywords}
lattice codes, relay channel, Gaussian relay channel, decode and forward, compress and forward
\end{keywords}


\section{Introduction}
\label{sec:intro}

The derivation of achievable rate regions for general networks including relays has classically used codewords and codebooks consisting of independent, identically generated symbols (i.i.d. random coding). Only in recent years have codes which possess additional structural properties, which we term structured codes, been used in networks with relays \cite{nazer2011compute,  Narayanan:2010, Nam:IEEE, nam:2009nested, Ozgur:2010:lattice, Nokleby:2011, noklebyunchaining}. The benefit of using structured codes in networks lies not only in a somewhat more constructive achievability scheme and possibly computationally more efficient decoding than i.i.d. random codes, but also in actual rate gains which exploit the structure of the codes --  their linearity in Gaussian channels -- to decode combinations of codewords rather than individual codewords / messages.
 While past work has focussed mainly on specific scenarios in which structured or lattice codes are particularly beneficial, missing is the demonstration that lattice codes may be used to achieve the same rate as known i.i.d. random coding based schemes in Gaussian relay networks, in addition to going above and beyond i.i.d. random codes in certain scenarios. In this work  we demonstrate generic nested lattice code based schemes with computationally more efficient lattice decoding for achieving the Decode-and-Forward and Compress-and-Forward rates in Gaussian relay networks which achieve at least the same rate regions as the corresponding rates achieved using Gaussian random codes. In the longer term,  these strategies may be combined with ones which exploit the linear structure of lattice codes to obtain structured coding schemes for arbitrary Gaussian relay networks. 
 Towards this goal, we illustrate how the DF based lattice scheme may be combined with strategies which exploit the linearity of lattice codes in two examples: the two-way relay channel with direct links and the multiple-access relay channel.
\subsection{Goal and motivation}

 In relay networks, as opposed to single-hop networks,  multiple links or routes exist between a given source and destination.   
 Of key importance in such networks is how to best jointly utilize these links, which -- in a single source scenario -- all carry the same message and effectively cooperate with each other to maximize the number of messages that may be distinguished. 
 The three node relay channel with one source with one message for one destination aided by one relay is the simplest relay network where pure cooperation between the links is manifested. Information may flow along the direct link or along the relayed link; how to manage or have these links cooperate to best transmit this message is key to approaching capacity for this channel. Despite this network's simplicity, its capacity remains unknown in general. However, the following two ``cooperative'' achievability schemes may approach capacity under specific channel conditions:  ``Decode-and-Forward'' (DF) and  ``Compress-and-Forward'' (CF) strategies described in \cite{Cover:1979, xie_kumar_04, Kramer:Gastpar:Gupta, Gammal:LN}.  In the DF scheme, 
 the receiver does not obtain the entire message from the direct link nor the relayed link. Rather, 
 cooperation between the direct and relayed links may be implemented by having the receiver decode a list of possible messages (or codewords)  from the direct link, another independent list from the coherent combination of the direct link and the relayed link, which it then intersects to obtain the message sent\footnote{There are alternative schemes for implementing DF, but the main intuition about combining information along two paths remains the same.}. 
 In the CF scheme of \cite{Cover:1979}, cooperation is implemented by a two-step decoding procedure combined with Wyner-Ziv binning. 
 
Generalizations of these i.i.d. random-coding based DF and CF schemes  have been proposed for general multi-terminal relay networks \cite{ xie_kumar_04,xie2005achievable,lim2011noisy}.  However, in recent years lattice codes have been shown to outperform random codes in several Gaussian multi-source network scenarios due to their  linearity property  \cite{nazer2011compute, bresler2010approximate, jafar:very_strong_IC,Narayanan:2010, Nam:IEEE, nam:2009nested}. As such, 
one may hope to derive a coding scheme which combines the best of both worlds, i.e. incorporate lattice codes with their linearity property into coding schemes for general Gaussian networks. 
At the moment we cannot simply replace i.i.d. random codes with lattice codes. That is, while nested lattice codes have been shown to be capacity achieving in the point-to-point Gaussian channel, in relay networks with multiple links/paths and the possibility of cooperation, technical issues need to be solved before one may replace random codes with lattice codes. 
 
  In this paper, we make progress in this direction by demonstrating lattice-based cooperative techniques for a number of relay channels.  One of the key new technical ingredients in the DF schemes is the usage of a lattice list decoding scheme to decode a list of lattice points (using lattice decoding) rather than a single lattice point. 
We  then extend this lattice-list-based cooperative technique and combine it with the linearity of lattice codes to provide gains for some channel conditions over i.i.d. random codes in scenarios with multiple cooperating links. 
\subsection{Related work} 

In showing that lattice codes may be used to replace i.i.d. random codes in Gaussian relay networks, we build upon work on   relay channels, on the existence of  ``good'' nested lattice codes for Gaussian source and channel coding, and on recent advancements in using lattices in multiple-relay and multiple-node scenarios. We outline the most relevant related work. 

{\bf Relay channels.} 
Two of our main results are the demonstration that nested lattice codes may be used to achieve the DF and CF rates achieved by random Gaussian codes \cite{Cover:1979}. For the DF scheme, we mimic the Regular encoding/Sliding window decoding DF strategy \cite{xie_kumar_04, Kramer:Gastpar:Gupta} in which the relay decodes the message of the source, re-encodes it, and then forwards it. The destination combines the information from the source and the relay by intersecting two independent lists of messages obtained from the source and relayed links respectively, over two transmission blocks. 
We will re-derive the DF rate, but with lattice codes replacing the random i.i.d. Gaussian codes.  
Of particular importance is constructing and utilizing a  lattice version of the list decoder. 
It is worth mentioning that the concurrent work \cite{Nokleby:2011} uses a different lattice coding scheme to achieve the DF rate in the three-node relay channel which does not rely on list decoding but rather on a careful nesting structure of the lattice codes. 

The DF scheme of \cite{Cover:1979} restricts the rate by requiring the relay to decode the message. The Compress-and-Forward (CF) 
achievability scheme of  \cite{Cover:1979} for the relay channel places no such restriction, as the . 
relay compresses its received signal and forwards the compression index. 
In Cover and El Gamal's original CF scheme, the relay's compression technique utilizes a form of binning related to the  Wyner-Ziv rate-distortion problem with decoder side-information \cite{Wyner:1974:WZcoding}. In \cite{Zamir:1998:WZ, Zamir:2002:binning} the authors describe a lattice version of  the noiseless quadratic Gaussian Wyner-Ziv coding scheme, where lattice codes quantize/compress the continuous signal; this will form the basis for our lattice-based CF strategy. 
Another simple structured approach to the relay channel is considered in \cite{khormuji:instantaneous, khormuji:sawtooth} where one-dimensional structured quantizers are used in the relay channel subject to instantaneous (or symbol-by-symbol) relaying. 

Our extension of the single relay DF rate to a multiple relay DF rate is based on the DF multi-level relay channel scheme presented in \cite{xie2005achievable, xie_kumar_04}. These papers essentially extend the DF rate of \cite{Cover:1979}; the central idea behind mimicking the scheme of \cite{xie_kumar_04,  xie2005achievable}  is the repeated usage of the lattice list decoder, enabling the message to again be decoded from the intersection of multiple independent lists formed at the destination from the different relay - destination links.

{\bf Lattice codes for single-hop channels.}  Lattice codes are known to be ``good'' for almost everything in Gaussian point-to-point, single-hop channels \cite{Erez:2004, erez2005lattices, zamir-lattices}, from both source and channel coding perspectives. 
In particular, nested lattice codes have been shown to be capacity achieving for the AWGN channel,  the AWGN broadcast channel \cite{Zamir:2002:binning} and the AWGN multiple access channel \cite{nazer2011compute}. 
Lattice codes may further be used in achieving the capacity of Gaussian channels with interference or state known at the transmitter (but not receiver) \cite{gelfand} using a lattice equivalent \cite{Zamir:2002:binning} of dirty-paper coding (DPC) \cite{costa}.  The nested lattice approach of  \cite{Zamir:2002:binning} for the dirty-paper channel is extended to dirty-paper networks in \cite{ Philosof:DirtyMAC},
 where in some scenarios lattice codes are interestingly shown to outperform random codes. 
In $K\geq 3$-user interference channels, their structure has enabled the decoding of  (portions of) ``sums of interference'' terms \cite{bresler2010approximate, jafar:very_strong_IC, jafarian2010gaussian, ordentlich:K}, allowing receivers to subtract off this sum rather than try to decode individual interference terms in order to remove them. From a source coding perspective,  lattices have been useful in distributed Gaussian source coding when reconstructing a linear function \cite{krithivasan2009lattices, wagner2011distributed}.  

{\bf Lattice codes in multi-hop channels.} 
The linearity property of lattice codes have been exploited in the Compute-and-Forward  framework \cite{nazer2011compute} for Gaussian multi-hop wireless relay 
networks  \cite{Narayanan:2010, Nam:IEEE, nam:2009nested}. There, intermediate relay nodes decode a linear combination, or 
equation,  of the transmitted codewords or equivalently messages by exploiting the noisy linear combinations provided by the channel.  Through the use of  nested lattice codes, it was 
shown that decoding linear combinations may be done at higher rates than decoding the individual codewords -- one of the key benefits of using structured 
rather than i.i.d. random codewords \cite{KornerMarton:1979}. 
Recently, progress has been made in characterizing the capacity of a single source, single destination, multiple relay network to within a constant gap for arbitrary network topologies \cite{avestimehr2011wireless}. Capacity was initially shown to be approximately achieved via an i.i.d. random quantize-map-and-forward based coding scheme \cite{avestimehr2011wireless} and alternatively, using an extension of CF based techniques termed ``noisy network coding'' \cite{lim2011noisy}. Recently, 
relay network capacity was also shown to be achievable using nested lattice codes for quantization and transmission \cite{Ozgur:2010:lattice}. Alternatively, using a new  ``computation alignment'' scheme which couples lattice codes in a compute-and-forward-like framework \cite{nazer2011compute} together with a signal-alignment scheme reminiscent of ergodic interference alignment \cite{nazer:ergodic}, the work \cite{niesen2011computation} was able to show a capacity approximation for multi-layer wireless relay networks with an approximation gap that is independent of the network depth. 
While lattices have been used in relay networks, the goals so far have mainly been to demonstrate their utility in specific networks in which decode linear combinations of messages is beneficial, or to achieve finite-gap results. 

As a first example of the use of lattices in multi-hop scenarios, we will consider the Gaussian two-way relay channel \cite{Nam:IEEE, Narayanan:2010}.  The two-way relay channel consists of three nodes: two terminal nodes 1 and 2  that wish to exchange their two independent messages through the help of one relay node R. 
When the  terminal nodes employ {\it nested lattice codes}, the sum of their signals is again a lattice point and may be decoded at the relay.
Having the relay send this sum (possibly re-encoded) allows the terminal nodes to exploit their own message side-information to recover the other user's message \cite{Narayanan:2010,Nam:IEEE}. Gains over DF schemes where both terminals transmit simultaneously to the relay 
stem from the fact that, if using random Gaussian codebooks, the relay will see a multiple-access channel and require the decoding of both individual messages, even though the sum is sufficient. In contrast, no multiple-access (or sum-rate) constraint is imposed by the lattice decoding of the sum. 
An alternative non-DF (hence no rate constraints at relay) yet still structured approach to the two-way relay channel is explored in \cite{khormuji:ISIT, fong}, where simple one dimensional structured quantizers are used for a symbol-by-symbol Amplify-and-Forward based scheme. 
In the two-way relay channel, models with and without direct links between the transmitters have been considered. 
While random coding techniques have been able to exploit both the direct link and relayed links, lattice codes have only been used in channels without direct links. Here, we will present a lattice coding scheme which will combine the linearity properties, leading to less restrictive decoding constraints at the relay, with direct-link information, allowing for a form of lattice-enabled two-way cooperation.

A second example in which we will combine the linearity property with direct-link cooperation is the Gaussian multiple-access relay channel  \cite{sankar2007offset, Kramer:Gastpar:Gupta, woldegebreal2007multiple}. In this model, two sources wish to communicate independent messages to a common destination with the help of a single relay. 
As in the Gaussian two-way relay channel, the relay may choose to decode the sum of the codewords  using lattice codes, rather than the individual codewords (as in random coding based DF schemes), which it would forward to the destination. The destination would combine this sum with direct-link information (cooperation). As in the two-way relay channel, decoding the sum at the relay eliminates the multiple access sum-rate constraint. 

\subsection{Contributions and outline}
Our contributions center around demonstrating that lattices may achieve the same rates as currently known Gaussian i.i.d. random coding-based achievability schemes for relay networks. While we do not prove this sweeping statement in general, we make progress towards this goal along the following lines:
\begin{itemize}
\item {\bf Preliminaries and Lattice List Decoder:} In Section \ref{sec:prelim} we briefly outline lattice coding preliminaries and notation before outlining key technical lemmas that will be needed,  including the central contribution of Section \ref{sec:prelim} -- the proposed Lattice List Decoding technique in Theorem \ref{thm:list}.
\item {\bf Decode-and-Forward, single source:} This Lattice List Decoding technique is used to show that nested lattice codes may achieve the Decode-and-Forward rate for the Gaussian relay channel achieved by i.i.d. random Gaussian codes  \cite{Cover:1979}  in Section \ref{sec:DF1}, Theorem \ref{thm:DF}. We furthermore extend this result to the general single source, multiple relay Gaussian channel in Theorem \ref{thm:DFm}.
\item  {\bf Decode-and-Forward, multiple source including two-way relay and multiple access relay channels:} In Section \ref{sec:DFm} relays  decode and forward combinations of messages as in the Compute-and-Forward framework, which is combined with direct link side-information at the destination. In particular, we present lattice-based achievable rate regions for the Gaussian two-way relay channel with direct links in Theorem \ref{thm:two-way}, and the Gaussian multiple-access relay channel in Theorem \ref{thm:MARC}.
\item {\bf Compress-and-Forward, single source:} In Section \ref{sec:CF1}, we revisit our goal of showing that lattice codes may mimic the performance of i.i.d. Gaussian codes in the relay channel by demonstrating a lattice code-based Compress-and-Forward scheme  which achieves the same rate as the CF scheme in  \cite{Cover:1979} evaluated for i.i.d. Gaussian codebooks. The proposed lattice CF scheme is based on a variation of the  lattice-based Wyner-Ziv scheme of  \cite{Zamir:1998:WZ, Zamir:2002:binning}, as outlined in Theorem \ref{thm:LCF}. We note that lattices have been shown to achieve the Quantize-Map-and-Forward rates for general relay channels using Quantize-and-Map scheme (similar to the CF scheme) which simply quantizes the received signal at the relay and re-encodes it without any form of binning / hashing in \cite{Ozgur:2010:lattice}; the contribution is to show an alternative lattice-coding based achievability scheme which employs computationally more efficient lattice decoding. 
\end{itemize}
%
%
%
%
%
%
%
%
%
%
%
%
%
%
\section{Preliminaries, Notation, and the Lattice List Decoder}
\label{sec:prelim}

We introduce our notation for lattice codes, nested lattice codes, and nested lattice chains and present several existing lemmas. We next present the new Lattice List Decoder (Theorem \ref{thm:list})  in which the decoder, instead of outputting a single estimated codeword, outputs a list which contains the correct one with high probability. The lemma bounds the number of points in the list.  The unique-decoding equivalent of the Lattice List Decoder Theorem \ref{thm:list} is provided in Lemma \ref{lem:unique}. 

\subsection{Lattice codes}

Our notation for (nested) lattice codes for transmission over AWGN channels  follows that of  \cite{Zamir:2002:binning, nam:2009nested}; comprehensive treatments may be found in \cite{loeliger1997averaging, Zamir:2002:binning,  Erez:2004} and in particular \cite{zamir-lattices}.  An $n$-dimensional lattice $\Lambda$ is a discrete subgroup of Euclidean space $\mathbb{R}^n$ with Euclidean norm $|| \cdot ||$ under vector addition and may be expressed as all integral combinations of basis vectors ${\bf g_i}\in {\mathbb R}^n$
\[ \Lambda = \{ \lambda = G \; {\bf i}: \; {\bf i}\in \mathbb{Z}^n\},\]
for $\mathbb{Z}$ the set of integers, $n\in {\mathbb Z}_+$, and $G := [{\bf g_1} | {\bf g_2}| \cdots {\bf g_n}]$ the $n\times n$ generator matrix corresponding to the lattice $\Lambda$. We use bold ${\bf x}$ to denote  column vectors, ${\bf x}^T$ to denote the transpose of the vector ${\bf x}$. All vectors are generally in ${\mathbb R}^n$ unless otherwise stated, and all logarithms are base 2. Let ${\bf 0}$  denote the all zeros vector of length $n$, ${\bf I}$ denote the $n\times n$ identity matrix, and ${\cal N}(\mu,\sigma^2)$ denote a Gaussian random variable (or vector) of mean $\mu$ and variance $\sigma^2$.  Define $C(x): = \frac{1}{2}\log_2\left(1+x \right)$. Further define

$\bullet$ The {\it nearest neighbor lattice quantizer} of $\Lambda$ as \[ Q({\bf x}) = \arg \min_{\lambda\in \Lambda} ||{\bf x}-\lambda||;\]

$\bullet$ The $\mod \Lambda$ operation as ${\bf x} \mod \Lambda : = {\bf x} - Q({\bf x})$;

$\bullet$ The {\it fundamental Voronoi region of $\Lambda$} as the points closer to the origin than to any other lattice point \[\mathcal{V}:= \{{\bf x}:Q({\bf x}) = {\bf 0}\},\]
which is of volume $V: = \mbox{Vol}({\mathcal V})$ (also sometimes denoted by $V(\Lambda)$ or $V_i$ for lattice $\Lambda_i$);

$\bullet$ The {\it second moment per dimension of a uniform distribution over ${\mathcal V}$} as
\[ \sigma^2(\Lambda) : = \frac{1}{V}\cdot \frac{1}{n} \int_{\mathcal V} ||{\bf x}||^2 \; d{\bf x};\]

$\bullet$ The {\it normalized second moment} of a lattice $\Lambda$ of dimension $n$ as
\[ G(\Lambda) : = \frac{\sigma^2(\Lambda)}{V^{2/n} };\]

$\bullet$ A sequence of $n$-dimensional lattices  $\Lambda^{(n)}$ is said to be {\it Poltyrev good} \cite{zamir_q, Erez:2004, nam:2009nested} (in terms of channel coding over the AWGN channel) if, for ${\bf \overline{Z}}\sim {\cal N}(0,\overline{\sigma}^2{\bf I})$ and $n$-dimensional vector, we have
\[ \Pr\{{\bf \overline{Z}} \notin {\mathcal V}^{(n)}\} \leq e^{-n (E_P(\mu) - o_n(1) )} , \]
 which upper bounds the error probability of nearest lattice point decoding when using lattice points as codewords in the AWGN channel. Here $E_p(\mu)$ is the {\it Poltyrev exponent }\cite{Erez:2004, poltyrev1994coding} which is given as 
\[ E_p(\mu)= \left\{ \begin{array}{rcl} \frac{1}{2} [(\mu -1) - \log \mu ], &  1< \mu \leq 2 \\ 
\frac{1}{2} \log \frac{e\mu}{4} &  2\leq\mu \leq4, \\ 
 \frac{\mu}{8} &  \mu \geq 4.
\end{array}\right.\]
and $\mu$ is volume-to-noise ratio (VNR)  defined as \cite{erez2005lattices}  \[\mu: = \frac{(\mbox{Vol}({\mathcal V}))^{2/n}}{2\pi e\overline{\sigma}^2} + o_n(1).\]
Since $E_p(\mu) > 0$ for $\mu > 1$, a necessary condition for the reliable decoding of a single point is $\mu>1$,  thereby relating the size of the fundamental Voronoi region (and ultimately how many points one can transmit reliably) to the noise power, aligning well with our intuition about Gaussian channels.


$\bullet$ {A sequence of $n$-dimensional lattices $\Lambda^{(n)}$ is said to be {\it Rogers good} \cite{rogers}} if 
\[ \lim_{n \rightarrow \infty} \frac{r_{cov}^{(n)}}{r_{eff}^{(n)}} = 1,  \]
where the covering radius $r_{cov}^{(n)}$ is the radius of the smallest sphere which contains the fundamental Voronoi region of $\Lambda^{(n)}$, and  
the effective radius $r_{eff}^{(n)}$ is the radius of a sphere of the same volume as the fundamental Voronoi  region of $\Lambda^{(n)}$.

$\bullet$ A sequence of $n$-dimensional lattices $\Lambda^{(n)}$ is said to be {\it good for mean-squared error quantization} if
\[ \lim_{n\rightarrow \infty} G(\Lambda^{(n)}) = \frac{1}{2\pi e};\]
It may be shown that if a sequence of lattices is Rogers good, that it is also good for mean-squared error quantization \cite{Erez:2005}. 
Furthermore, for a Rogers' good lattice $\Lambda$, it may be shown that $\sigma^2(\Lambda)$ and $V = \mbox{Vol}({\mathcal V})$ are in one-to-one correspondence (up to a constant) as in \cite[Appendix A]{nam:2009nested}; hence for a Rogers good lattice we may define either its second moment per dimension or its volume. This will be used in generating nested lattice chains.


{Finally, we include a statement of the useful ``Crypto lemma'' for completeness.}
{\begin{lemma}{{\it Crypto lemma \cite{Erez:2004,Forney:ShannonWinener:2003}.}}
\label{lem:crypto}
For any random variable ${\bf x}$ distributed over the fundamental region $\mathcal{V}$ and statistically independent of ${\bf U}$, which is  uniformly distributed over $\mathcal{V}$, 
 $({\bf  x} + {\bf U}) \mod \Lambda$ is independent of ${\bf x}$ and uniformly distributed over $\mathcal{V}$.
\end{lemma}}

\subsection{Nested lattice codes}
\label{subsec:nested}

Consider two lattices $\Lambda$ and $\Lambda_c$ such that $\Lambda \subseteq \Lambda_c$ with fundamental regions ${\cal V}, {\cal V}_c$ of volumes $V, V_c$ respectively. Here $\Lambda$ is termed the {\it coarse} lattice which is a sublattice of  $\Lambda_c$,  the {\it fine} lattice, and hence $V \geq V_c$.  When transmitting over the AWGN channel, one may use the set  $ \mathcal{C}_{\Lambda_c, {\cal V}} = \{ \Lambda_c \cap \mathcal{V} \} $ as the codebook. The coding rate $R$ of this {\it nested  $(\Lambda, \Lambda_c)$ {lattice pair}} is defined as
\[ R = \frac{1}{n} \log |\mathcal{C}_{\Lambda_c, {\cal V}}| = \frac{1}{n} \log \frac{V}{V_c},\] where $\rho = |\mathcal{C}_{\Lambda_c, {\cal V}}|^{\frac{1}{n}} = \left( \frac{V}{V_c} \right)^{\frac{1}{n}}$ is the nesting ratio of the {nested lattice pair}.  It was shown that there exist nested lattice pairs which achieve the  capacity of the AWGN channel \cite{Erez:2004}.

\subsection{Nested lattice chains}
\label{subsec:chains}

In the following, we will use an extension of nested lattice codes termed nested lattice chains as in \cite{nam:2009nested,  Nam:IEEE}, and shown  in Figure \ref{fig:chain} (chain of length 3).
We first re-state  a slightly modified version of \cite[Theorem 2]{nam:2009nested} on the existence of good nested lattice chains, of use in our achievability proofs.

\smallskip

\theorem{{\it Existence of ``good'' nested lattice chains (adapted from Theorem 2 of \cite{nam:2009nested}).}
\label{thm:nam}
For any $P_1 \geq P_2 \geq \dots \geq P_{K} >0$ and $\gamma > 0$, there exists a sequence of $n$-dimensional lattice $\Lambda_1 \subseteq \Lambda_2 \subseteq \dots \subseteq \Lambda_{K} \subseteq \Lambda_C$ ($\mathcal{V}_1 \supseteq \mathcal{V}_2 \supseteq \dots \supseteq \mathcal{V}_K \supseteq \mathcal{V}_C$) satisfying:\\
a) $\Lambda_1$, $\Lambda_2$, $\dots$, $\Lambda_{K}$ are simultaneously Rogers-good and and Poltyrev-good while $\Lambda_C$ is Poltyrev-good. \\
b) For any $\delta > 0$, $P_i-\delta \leq \sigma^2(\Lambda_i) \leq P_i $, $1 \leq i \leq K$ for sufficiently large n.\\
c) The coding rate associated with {the nested lattice pair }$\Lambda_{K} \subseteq \Lambda_{C}$ is $R_{K,C} = \frac{1}{n} \log \frac{V_{K}}{V_C} = \gamma + o_n(1)$ where $o_n(1) \rightarrow 0$ as $n \rightarrow \infty$.
 Moreover, for $1 \leq i  \, <  \,  j \leq K$, the coding rate of the nested lattice pair $\Lambda_i \subseteq \Lambda_j$ is $R_{i, j} := \frac{1}{n} \log \frac{V_i}{V_j} = \frac{1}{2} \log \frac{P_i}{P_j} + o_n(1)$ and $R_{i,C} = R_{i,K} + R_{K,C} = \frac{1}{2} \log \frac{P_i}{P_{K}} + \gamma + o_n(1)$ ($1\leq i \leq K-1$). 
\begin{proof}
From Theorem 2 of \cite{nam:2009nested} there exists a nested lattice chain which satisfies the properties a) and b) and for which $R_{K,C} = \gamma + o_n(1)$, and $R_{i,C} = \frac{1}{n} \log \frac{V_i}{V_C} = R_{K,C} + \frac{1}{2}\log \frac{P_i}{P_K} + o_n(1)$. Now notice that $R_{i,j} = \frac{1}{n} \log \frac{V_i}{V_j} = \frac{1}{n} \log \frac{V_i}{V_C} - \frac{1}{n} \log \frac{V_C}{V_j} = R_{i,C} - R_{j,C} = \frac{1}{2}\log\frac{P_i}{P_j} + o_n(1)$.
\end{proof}
}

\begin{figure}
\centering
\includegraphics[width=6cm]{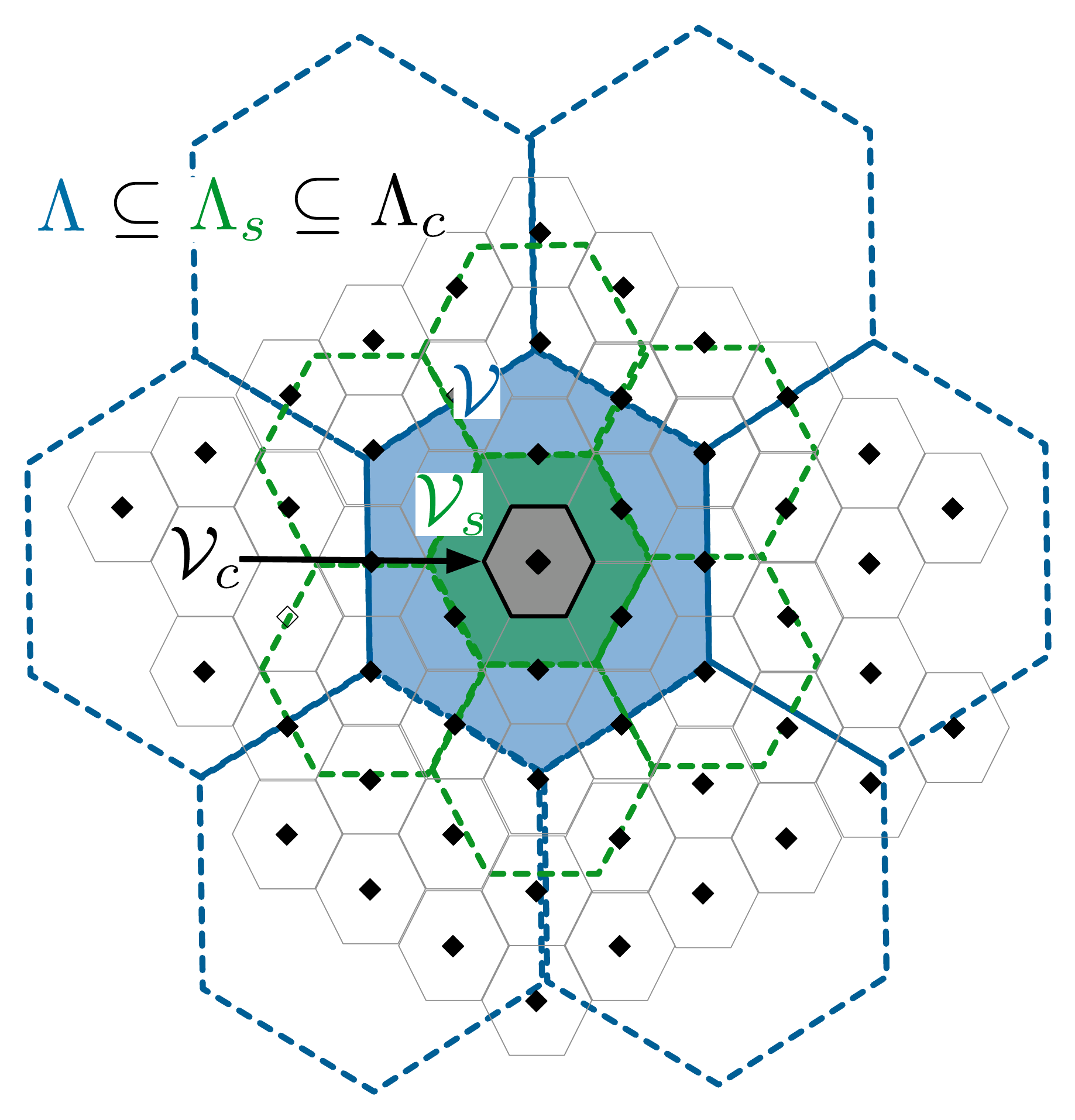}
\caption{A lattice chain $\Lambda\subseteq\Lambda_s \subseteq \Lambda_c$ with fundamental regions  $\mathcal{V}  \supseteq \mathcal{V}_s \supseteq \mathcal{V}_c $ of volumes $V\geq V_s\geq V_c$. Color is useful.}
\label{fig:chain}
\end{figure}

\subsection{A lattice list decoder}
\label{subsec:list}

List decoding here refers to a decoding procedure in which, instead of outputting a single codeword corresponding to a single message, the decoder outputs a list of possible codewords which includes the correct (transmitted) one {with high probability}. Such a decoding scheme is useful in cooperative scenarios when a message is transmitted above the capacity of a given link (and hence the decoder would not be able to correctly distinguish the true transmitted codeword from that  given link), and is combined with additional information at the receiver to decode a single message point from within the list.
We present our key theorem next which bounds the list size for a lattice list decoder which will decode a list which contains the correct message with high probability. 



\smallskip


\smallskip
\theorem{{\it Lattice list decoding in mixed noise.}
\label{thm:list}
Consider the channel ${\bf Y} = {\bf X}+{\bf Z}$, subject to input power constraint $\frac{1}{n}E[{\bf X}^T{\bf X}]\leq P$, where ${\bf Z} = {\bf Z_G} + \sum_{i=1}^L {\bf Z_{i}}$ is noise which is a mixture of Gaussian noise ${\bf Z_G} \sim {\cal N}({\bf 0}, \sigma_G^2 {\bf I})$ and independent noises ${\bf Z_{i}}$ which are
uniformly distributed over fundamental Voronoi  regions of Rogers-good lattices with second moments $P_{i}$. Thus, ${\bf Z}$ is of equivalent total variance $N = \frac{1}{n}\mathbb{E}({\bf Z}^T {\bf Z}) = \sigma_G^2 + \sum_{i=1}^L P_i$.
For any $|L| > 2^{n(R - C(P/N))}, \, \delta > 0$, $R > C(P/N)$,  and $n$ large enough, there exists
a chain of nested lattices such that the lattice list decoder can produce
a list of size $|L|$, which does not contain the correct codeword with
probability smaller than $\delta$. 
}

\begin{proof}

{\bf Encoding:} We consider a good nested lattice chain  $\Lambda\subseteq\Lambda_s \subseteq \Lambda_c$ as in Figure \ref{fig:chain}and  Theorem \ref{thm:nam}, in which $\Lambda$ and $\Lambda_s$ are both Rogers good and Poltyrev good while $\Lambda_c$ is Poltyrev good. We define the coding rate $R = \frac{1}{n} \log \frac{V}{V_c}$ and the nesting rate $R_1 = \frac{1}{2} \log \frac{V}{V_s}$.  
Each message $w\in \{1,\dots,2^{nR}\}$ is one-to-one mapped to the lattice point ${\bf t}(w) \in \mathcal{C}_{\Lambda_c, {\cal V}}= \{ \Lambda_c \cap \mathcal{V} \} $, and the transmitter sends ${\bf X} = ({\bf t}(w) - {\bf U}) \mod \Lambda$, where ${\bf U}$ is an $n$-dimensional dither signal (known to the encoder and decoder) uniformly distributed over $\mathcal{V}$.

{\bf Decoding:} Upon receiving {\bf Y}, the receiver computes
\begin{align}
{\bf Y^\prime} &= (\alpha {\bf Y} + {\bf U}) \mod \Lambda  \nonumber \\
&= ({\bf t}(w) - (1-\alpha){\bf X} + \alpha {\bf Z}) \mod \Lambda \nonumber \\
&= ({\bf t}(w) + (- (1-\alpha){\bf X} + \alpha {\bf Z} )\mod \Lambda ) \mod \Lambda \nonumber \\
&= ({\bf t}(w) + {\bf Z^\prime}) \mod \Lambda, \label{eq:YY}
\end{align}
for $\alpha\in \mathbb{R}$. We choose $\alpha$ to be the MMSE coefficient $\alpha = \frac{P}{P+N}$ and note that the equivalent noise ${\bf Z^\prime} = (-(1-\alpha){\bf X} + \alpha {\bf Z} ) \mod \Lambda$ is independent of ${\bf t}(w)$. 
The receiver decodes the {\it list}  of  messages
\begin{equation}
L_{S-D}^{w}({\bf Y}) := \{w | \; {\bf {t}}(w) \in  S_{\mathcal{V}_s,\Lambda_c} ({\bf Y^\prime}) \mod \Lambda\},
\label{eq:listdecoding}\end{equation}
where  \[ S_{\mathcal{V}_s,\Lambda_c}({\bf Y'}) := \bigcup_{\lambda_c \in \Lambda_c} \{\lambda_c|\lambda_c  \in ({\bf Y^\prime} + \mathcal{V}_s)\},\] 
is the set of lattice points $\lambda_c \in \Lambda_c$ inside ${\cal V}_s$ centered at the point ${\bf Y'}$ as shown in Figure \ref{fig:encdec}. 
\begin{remark}
The notation used for the list of messages, i.e. $L_{S-D}^{w}({\bf Y})$ should be understood as follows: the $S-D$ subscript is meant to denote the transmitter $S$ and the receiver $D$, the dependence on ${\bf Y}$ (rather than ${\bf Y'}$) is included, though in all cases we will make the analogous transformation from ${\bf Y}$ to ${\bf Y'}$ as in \eqref{eq:YY} (but for brevity do not include this in future schemes), and the superscript $w$ is used to recall what messages are in the list,  useful in multi-source and Block Markov schemes. 
\end{remark}

\smallskip

{\bf Probability of error for list decoding:} Pick $\delta>0$. In decoding a list, we require that the correct, transmitted codeword ${\bf t}(w)$ lies in the list with high probability as $n\rightarrow \infty$, i.e. the probability of error is (for $n$ the blocklength or dimension of the lattices) $P_{n,e} : = \Pr\{w \notin L_{S-D}^{w}({\bf Y}) | w\mbox{ sent}\}$, which should be made less than $\delta$ as $n\rightarrow \infty$. This is easy to do with large list sizes; we bound the list size next. 
The following Lemma allows us to more easily bound the probability of list decoding error.

\begin{figure*}
\centering
\includegraphics[width=14cm]{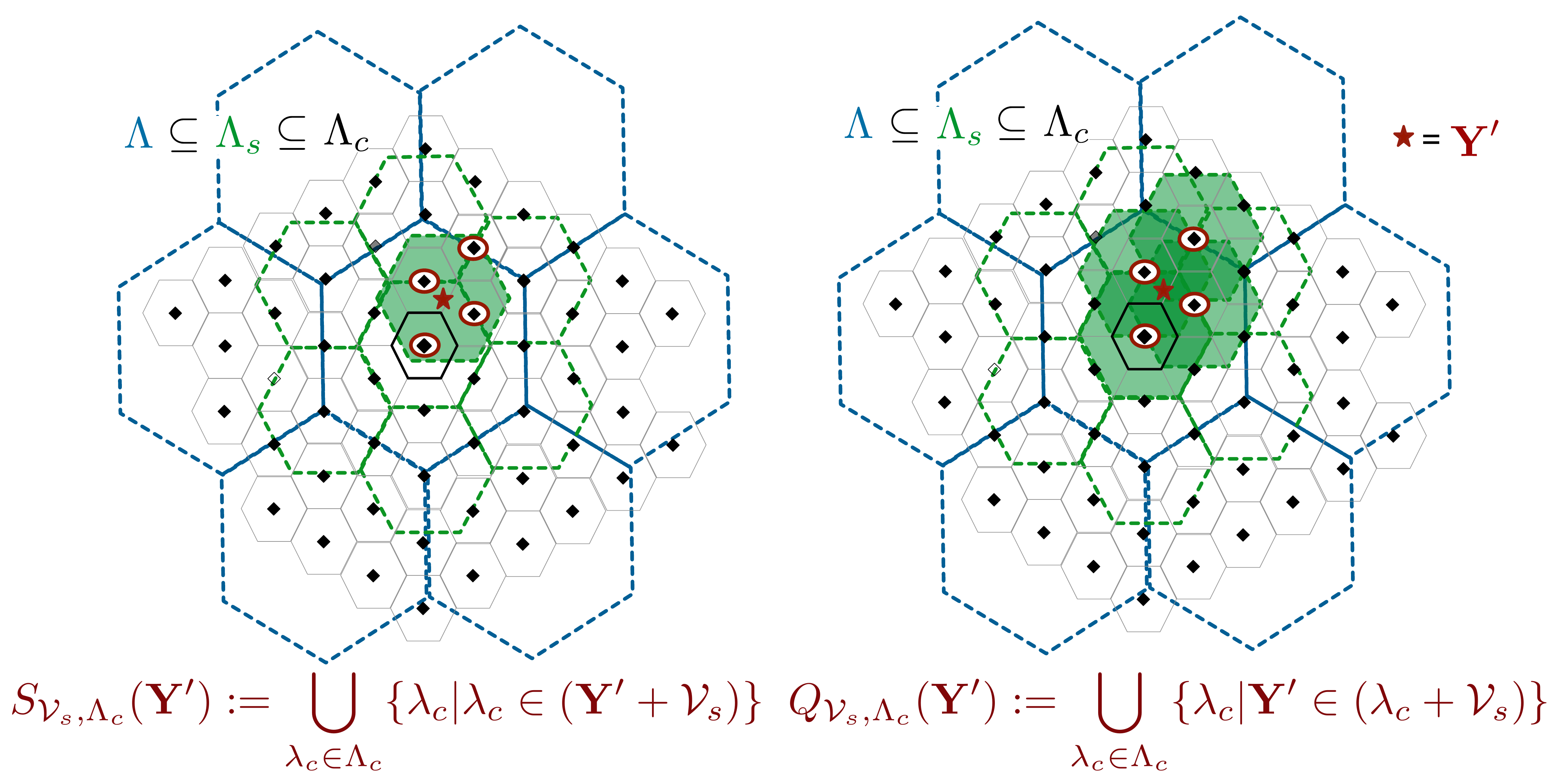}
\vspace{-0.2cm}
\caption{The two equivalent lists, in this example consisting of the four points encircled in red. 
Color is useful.}
\label{fig:encdec}
\end{figure*}


\begin{lemma}{{\it Equivalent decoding list.}}
\label{lem:equivalent}
For the nested lattices $\Lambda_s \subseteq \Lambda_c$ and given ${\bf Y'}\in \mathbb{R}^n$, define 
\begin{equation}
 Q_{\mathcal{V}_s,\Lambda_c} ({\bf Y^\prime}) := \bigcup_{\lambda_c \in \Lambda_c} \{\lambda_c|{\bf Y^\prime} \in (\lambda_c + \mathcal{V}_s) \}.\label{eq:listdecoding2}
 \end{equation}
 and 
 \[ S_{\mathcal{V}_s,\Lambda_c}({\bf Y'}) := \bigcup_{\lambda_c \in \Lambda_c} \{\lambda_c|\lambda_c  \in ({\bf Y^\prime} + \mathcal{V}_s)\},\] 
Then the sets $S_{\mathcal{V}_s,\Lambda_c} ({\bf Y^\prime}) \mod \Lambda$ and $Q_{\mathcal{V}_s,\Lambda_c} ({\bf Y^\prime}) \mod \Lambda$  are equal.
 \end{lemma}
 
\begin{proof}
$Q_{\mathcal{V}_s,\Lambda_c} ({\bf Y^\prime}) $ is the set of  $\lambda_c \in \Lambda_c$ points satisfying ${\bf Y^\prime} \in (\lambda_c + \mathcal{V}_s)$. 
Also note that the fundamental Voronoi region $\mathcal{V}$ of any lattice $\Lambda$ is centro-symmetric ($\forall x \in \mathcal{V}$, we have that $-x \in \mathcal{V}$) by definition of a lattice and fundamental Voronoi region (alternatively, see  \cite{Coppel}). Hence,  for any two points $x$ and $x^\prime$, and a centro-symmetric region $\mathcal{V}$, 
$x^\prime \in x + \mathcal{V} \Leftrightarrow x \in x^\prime + \mathcal{V}$. Applying this to $ S_{\mathcal{V}_s,\Lambda_c}({\bf Y'})$ and $Q_{\mathcal{V}_s,\Lambda_c} ({\bf Y^\prime}) $ yields the lemma.
\end{proof}

 We continue with the proof of Theorem \ref{thm:list}. We first use Lemma \ref{lem:equivalent} to see that the lists $S_{\mathcal{V}_s,\Lambda_c} ({\bf Y^\prime}) \mod \Lambda$ and $Q_{\mathcal{V}_s,\Lambda_c} ({\bf Y^\prime}) \mod \Lambda$ are equal. Next  notice that 
 the probability of error may be bounded as follows: 
\begin{align}
P_{n,e} & = \Pr\{w \notin L_{S-D}^{w}({\bf Y}) | \; w\mbox{ sent}\} \\
& = \Pr\{{\bf {t}}(w) \not \in  S_{\mathcal{V}_s,\Lambda_c} ({\bf Y^\prime}) \mod \Lambda | \; w\mbox{ sent}\}\\
& =  \Pr\{{\bf t}(w) \not \in  Q_{\mathcal{V}_s,\Lambda_c} ({\bf Y^\prime}) \mod \Lambda | \; w\mbox{ sent}\} \\
& = \Pr\{{\bf Y'} \not \in ({\bf t}(w) + \mathcal{V}_s)| \; w\mbox{ sent} \} \\
& = \Pr\{({\bf t}(w) + {\bf Z'}) \mod \Lambda \not \in ({\bf t}(w)+\mathcal{V}_s)| \; w\mbox{ sent}\} \\
& = \Pr\{ {\bf Z'} \not \in \mathcal{V}_s| \; w\mbox{ sent} \} \\
& \leq \Pr\{ {\bf Z''} \not \in \mathcal{V}_s| \; w\mbox{ sent} \}
\end{align}
where ${\bf Z^\prime} = (-(1-\alpha){\bf X} + \alpha {\bf Z} ) \mod \Lambda$ and ${\bf Z''} = - (1-\alpha){\bf X} + \alpha {\bf Z}$.
We now use Lemma \ref{lem:gaussianbound} to show that the pdf of ${\bf Z''}$ can be upper bounded by the pdf of a Gaussian random vector of not much larger variance, which in turn is used to bound the above probability of error.



\begin{lemma}
\label{lem:gaussianbound}
Let ${\bf Z}_G \sim {\cal N}({\bf 0}, \sigma_G^2{\bf I})$, ${\bf X}$ be uniform over the fundamental Voronoi region of the Rogers good $\Lambda$,  of effective and covering radii $r_{eff}$ and $r_{cov}$ and second moment $P$, and ${\bf Z}_i$  be uniform over the fundamental Voronoi region of the Rogers good $\Lambda_i$ of effective and covering radii $r_{eff,i}$ and $r_{cov,i}$ and second moments $P_i$, $i=1,\cdots L$. Let ${\bf Z''} : = -(1-\alpha){\bf X} + \alpha {\bf Z}_G + \alpha \sum_{i=1}^L {\bf Z}_i$.  
Then there exists an i.i.d. Gaussian vector 
\[ {\bf Z}^\star = -(1-\alpha) {\bf Z}^\star_X + \alpha {\bf Z}_G + \alpha \sum_{i=1}^L {\bf Z}^\star_i \] 
 with variance $\sigma^2$ satisfying 
\[\sigma^2 \leq (1-\alpha)^2 \left( \frac{r_{cov}}{r_{eff}} \right)^2 P + \alpha^2 \sigma^2_G + \alpha^2 \sum_{i=1}^L \left( \frac{r_{cov,i}}{r_{eff,i}} \right)^2 P_i \]
such that the density of ${\bf Z''}$ is upper bounded as:
\begin{equation} f_{\bf Z''} ({\bf z}) \leq e^{(c(n)+\sum_{i=1}^L c_i(n))n } f_{{\bf Z}^\star}({\bf z}) \label{eq:sumR} \end{equation}
where $c(n) = \ln \left( \frac{r_{cov}}{r_{eff}}\right) + \frac{1}{2} \ln 2\pi e G^{(n)}_{\mathcal{B}} + \frac{1}{n} $ and $c_i(n) = \ln \left( \frac{r_{cov,i}}{r_{eff,i}}\right) + \frac{1}{2} \ln 2\pi e G^{(n)}_{\mathcal{B}} + \frac{1}{n} $, and $G^{(n)}_{\mathcal{B}}$ is the normalized second moment of an $n$-dimensional ball.
 \end{lemma}
\begin{proof}
The proof follows \cite[Appendix A]{nazer2011compute} and \cite[Lemma 6 and 11]{Erez:2004} almost exactly, where the central difference with \cite[Appendix A]{nazer2011compute} is that we need to bound the pdf of a sum of random variables uniformly distributed over {\it different} Rogers good lattices rather than identical ones. This leads to the summation in the exponent of \eqref{eq:sumR} but note that we will still have $c(n), c_i(n) \rightarrow 0$ as $n\rightarrow \infty$. 
\end{proof}

Continuing the proof of Theorem \ref{thm:list},  according to Lemma \ref{lem:gaussianbound}, 
\begin{align} 
P_{n,e}  \leq \Pr\{ {\bf Z''} \not \in \mathcal{V}_s \} \leq e^{(c(n)+\sum_{i=1}^L c_i(n))n } \Pr\{ {\bf Z}^\star \not \in \mathcal{V}_s \} \label{eq:x}. 
\end{align}
To bound  $\Pr\{ {\bf Z}^\star \not \in \mathcal{V}_s \}$, 
we first need to show that the VNR of $\Lambda_s$ relative to ${\bf Z}^\star$, $\mu$, is greater than one:
\begin{align}
\mu = \frac{(V(\Lambda_s))^{2/n}}{2 \pi e \sigma^2 } + o_n(1)
&\geq \frac{(V(\Lambda))^{2/n}/2^{2R_1}}{2 \pi e \frac{P N}{P + N}} + o_n(1) \label{eq:1} \\
&= \frac{1}{2^{2R_1}} \frac{1}{2 \pi e G(\Lambda)} \frac{P}{ \frac{P N}{P + N}} + o_n(1) \label{eq:2}\\
&= \frac{1}{2^{2R_1}} \left( 1 + \frac{P}{N}\right) + o_n(1)  \label{eq:3} \\
&= 2^{2(C(P/N)-R_1)} + o_n(1) \label{eq:4}
\end{align}
where (\ref{eq:1}) follows from Lemma \ref{lem:gaussianbound},  the fact that $\Lambda$ and $\Lambda_i$ ($1 \leq i \leq L$) are all Rogers good, and recalling that $\alpha = \frac{P}{P+N}$, where $N = \sigma_G^2+ \sum_{i=1}^L P_i$. Then  (\ref{eq:2}) follows from the definition of $G(\Lambda)$ and $(\ref{eq:3})$ follows as $\Lambda$ is Rogers good. 
Combining (\ref{eq:x}),   (\ref{eq:4}), and  the fact that $\Lambda_s$ is Poltyrev good, by definition 
\begin{align}
P_{n,e}  &\leq e^{(c(n)+\sum_{i=1}^L c_i(n))n } \Pr\{ {\bf Z}^\star \not \in \mathcal{V}_s \} \\
&\leq e^{(c(n)+\sum_{i=1}^L c_i(n))n} e^{-n ( E_p(\mu) - o_n(1) )}\\
&\leq e^{-n(E_p(2^{2(C(P/N) - R_1)}) - o_n(1))} \label{eq:cn}
\end{align}
where \eqref{eq:cn} follows as $\Lambda, \Lambda_1, \cdots \Lambda_L$ are Rogers good and hence $c(n), c_i(n)$ all tend to $0$ as $n \rightarrow \infty$. 

To ensure $P_{n,e} <\delta $ as $n \rightarrow \infty $ we need
$ C(P/N) - R_1 > 0$, where $R_1 = \frac{1}{n} \log ( \frac{V}{V_s} ) = \frac{1}{2} \log (\frac{P}{P_s}) + o_n(1)$, and $n$ sufficiently large.
By choosing an appropriate $P_s$ according to Theorem \ref{thm:nam}, we may set $R_1 = \frac{1}{n} \log ( \frac{V}{V_s} ) = C(P/N) - \epsilon_n$ for any $\epsilon_n > 0$. 
Combining these, we obtain
\begin{equation}
 V_s = \left(\frac{N}{P+N}\right)^{n/2} 2^{n\epsilon_n} V.
 \label{eq:Vs}
 \end{equation}

The cardinality of the decoded list $L_{S-D}^w({\bf Y})$, in which the true codeword lies with high probability as $n\rightarrow \infty$, may be bounded as
\begin{align*}
|L_{S-D}^w({\bf Y})| &= \frac{V_s}{V_c} = \frac{\frac{N^{n/2}V}{(P+N)^{n/2}} 2^{n\epsilon_n}}{\frac{ V}{2^{nR}}} = 2^{n(R-C(P/N))}2^{n\epsilon_n},
\end{align*}
 since $ R = \frac{1}{n} \log ( \frac{V}{V_c} )$.
Setting $\epsilon_n = \frac{1}{n^2}$, $2^{n\epsilon_n} \rightarrow 1$, and so $|L_{S-D}^w({\bf Y})| \rightarrow 2^{n(R-C(P/N))}$ as $n \rightarrow \infty$.

\end{proof}

\begin{remark}
Note that in our Theorem statement we have assumed $R > C(P/N)$; when $R < C(P/N)$, the decoder can decode an unique codeword with high probability, as stated in Lemma \ref{lem:unique}. 

\lemma{{\it Lattice unique decoding in mixed noise.}
\label{lem:unique}
Consider the channel ${\bf Y} = {\bf X}+{\bf Z}$, subject to input power constraint $\frac{1}{n}E[{\bf X}^T{\bf X}]\leq P$, where ${\bf Z} = {\bf Z_G} + \sum_{i=1}^L {\bf Z_{i}}$ is noise which is a mixture of Gaussian noise ${\bf Z_G} \sim {\cal N}({\bf 0}, \sigma_G^2 {\bf I})$ and independent noises ${\bf Z_{i}}$ which are
uniformly distributed over fundamental Voronoi  regions of Rogers-good lattices with second moments $P_{i}$. Thus, ${\bf Z}$ is of equivalent variance $N = \frac{1}{n}\mathbb{E}({\bf Z}^T {\bf Z}) = \sigma_G^2 + \sum_{i=1}^L P_i$.
For any  $\delta > 0$, $R < C(P/N)$,  and $n$ large enough, there exist lattice codebooks such that the decoder can decode an unique codeword with probability  of error smaller than $\delta$. 
\begin{proof}
This lemma can be derived as a special case of Compute-and-Forward \cite[Theorem 1] {nazer2011compute}; in particular this is found in \cite[Example 2]{nazer2011compute}, where the decoder is interested in one of the messages and treats all other messages as noise. We may view ${\bf Z}_i$ in this lemma as the signals from other (lattice-codeword based) transmitters in \cite[Example 2]{nazer2011compute}.
\end{proof}}
\end{remark}
\section{Single source Decode and Forward}
\label{sec:DF1}

We first show that nested lattice codes may be used to achieve the Decode-and-Forward (DF) rate of  \cite[Theorem 5]{Cover:1979} for the Gaussian relay channel using nested lattice codes at the source and relay, and a lattice list decoder at the destination. We then extend this result to show that the generalized DF rate for a Gaussian relay network with a single source, a single destination and multiple DF relays may also be achieved using an extension of the single relay lattice-based achievability scheme.

\subsection{DF for the AWGN single relay channel}
\label{subsec:DF}

Consider a  relay channel in which the source node $S$, with channel input $X_{S}$  transmits a message $w\in \{1,2,\cdots, 2^{nR}\}$ to destination node $D$ which has access to the channel output $Y_{D}$ and is aided by a relay node $R$ with channel input and output $X_{R}$ and $Y_{R}$.  Input and output random variables lie in ${\mathbb R}$. At each channel use, the channel inputs and outputs are related as
$ Y_D = X_S+X_R+Z_D, \; Y_R = X_S+Z_R$,
where $Z_R, Z_D$ are independent Gaussian random variables of zero mean and variance $N_R$ and $N_D$ respectively. Let ${\bf X_S}$ denote a sequence of $n$ channel inputs (a row vector), and similarly, let ${\bf X_R}, {\bf Y_R}, {\bf Y_D}$ all denote the length $n$ sequences of channel inputs and outputs. Then the channel may be described by 
\begin{align}
{\bf Y_{D}} = {\bf X_{S}} + {\bf X_{R}} + {\bf Z_{D}},  \;\;\;\; {\bf Y_{R}} = {\bf X_{S}} + {\bf Z_{R}}, \label{eq:DFmodel}
\end{align}
where  ${\bf Z_D} \sim{\cal N}({\bf 0},N_D{\bf I})$ and ${\bf Z_R}\sim {\cal N}({\bf 0},N_R{\bf I})$, and inputs are subject to the power constraints  
$\frac{1}{n} E[{\bf X_S}^T{\bf X_S}] \leq P$ and $\frac{1}{n} E[{\bf X_R}^T{\bf X_R}]\leq P_R$. 

An $(2^{nR}, n)$ code for the relay channel consists of the set of messages $w$ uniformly distributed over  ${\cal M} : = \{1,2,\cdots 2^{nR}\}$, an encoding function $X_S^n: {\cal M} \rightarrow {\mathbb R}^n$ satisfying the power constraint, a set of relay functions $\{f_i\}_{i=1}^n$ such that the relay channel input at time $i$ is a function of the previously received relay channel outputs from channel uses $1$ to $i-1$, $X_{R,i} = f_i(Y_{R,1}, \cdots Y_{R,i-1})$, and finally a decoding function $g: {\cal Y}_D^n \rightarrow {\cal M}$ which yields the message estimate $\hat{w}: = g(Y_D^n)$. We define the average probability of error of the code to be $P_{n,e} : = \frac{1}{2^{nR}} \sum_{w\in {\cal M}} \Pr\{\hat{w} \neq w|w \mbox{ sent}\}$. The rate $R$ is then said to be achievable by a relay channel if, for any $\epsilon>0$ and for sufficiently large $n$, there exists an $(2^{nR},n)$ code such that $P_{n,e} < \epsilon$. The capacity $C$ of the relay channel is the supremum of the set of achievable rates. 


\begin{figure*}
\centering
\includegraphics[width=14cm]{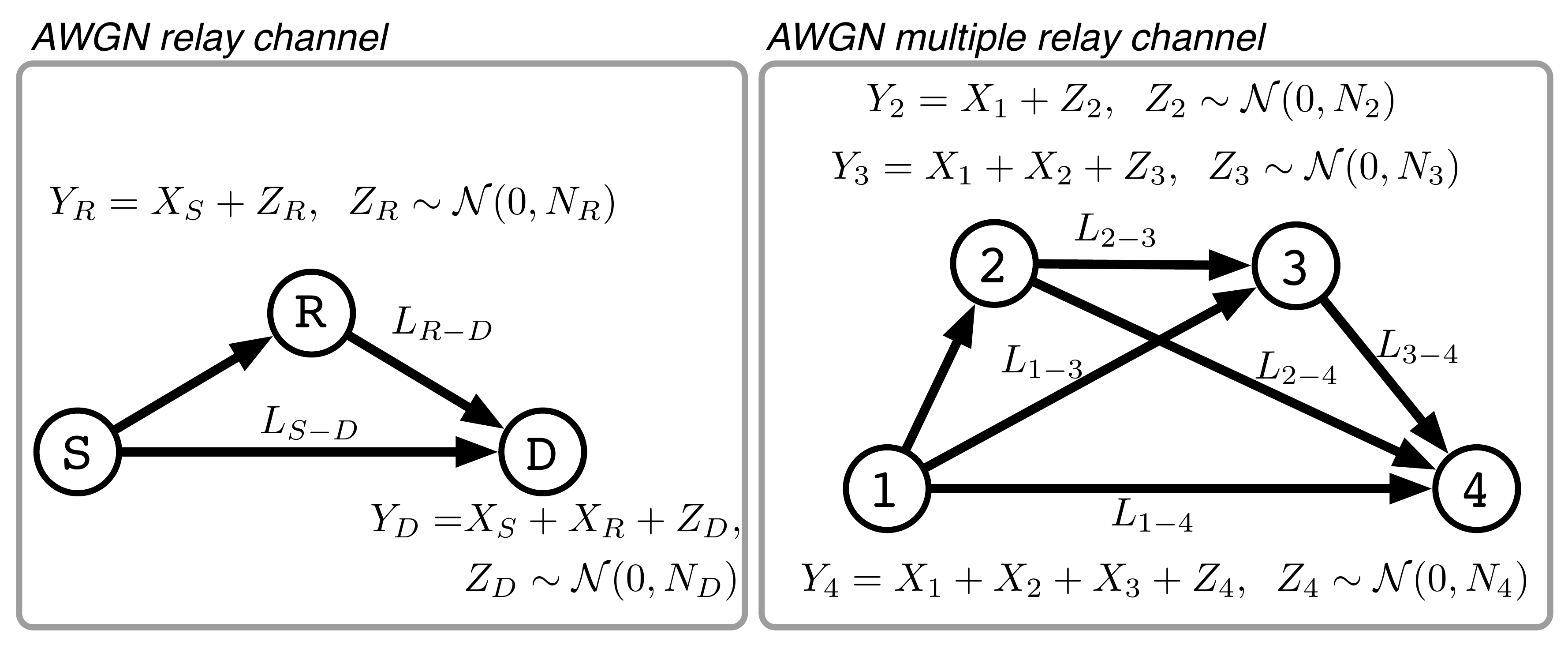}
\caption{The two Gaussian relay channels under consideration in Section \ref{subsec:DF} and Section  \ref{subsec:TWRC}. For the AWGN relay channel we have assumed a particular relay order (2,3) for our achievability scheme  and shown the equivalent channel model used in deriving the achievable rate rather than the general channel model. 
}\label{fig:relay-channel}
\end{figure*}

We are first interested in showing that the DF rate achieved by Gaussian random codebooks of \cite[Theorem 5]{Cover:1979} may be achieved using lattice codes.  As outlined in \cite{Kramer:Gastpar:Gupta}, this DF rate may be achieved using irregular encoding / successive decoding as in \cite{Cover:1979}, regular encoding / sliding-window decoding as first shown in \cite{carleial1982multiple}, and using regular encoding / backwards decoding as in \cite{Willems:thesis}. 
  We will mimic  the regular encoding/sliding-window decoding scheme of \cite{xie2005achievable}, which 
includes: (1) random coding, (2) list decoding, (3) two joint typicality decoding steps, (4) coding for the cooperative multiple-access channel, (5) superposition coding and (6) block Markov encoding. We re-derive the DF rate,  following the achievability scheme of  \cite{xie2005achievable}, 
but with lattice codes replacing the random Gaussian coding techniques.  Of particular importance is the usage of two lattice list decoders to replace two joint typicality decoding steps in the random coding achievability scheme.

\medskip

\begin{theorem}
\label{thm:DF}
{\it Lattices achieve the DF rate achieved by random Gaussian codebooks for the relay channel.} The following Decode-and-Forward rates can be achieved using nested lattice codes for the Gaussian relay channel described by \eqref{eq:DFmodel}:
\begin{equation}
R < \max_{\alpha\in [0,1]}\min\left\{\frac{1}{2}\log\left( 1+\frac{\alpha P}{N_R}\right), \frac{1}{2}\log\left(1+\frac{P+P_R+2\sqrt{\bar{\alpha}P P_R}}{N_D} \right)  \right\},\quad \bar{\alpha} = 1-\alpha.
\label{eq:DF}
\end{equation}
\end{theorem}
\begin{proof}

\begin{figure*}
\centering
\includegraphics[width=14cm]{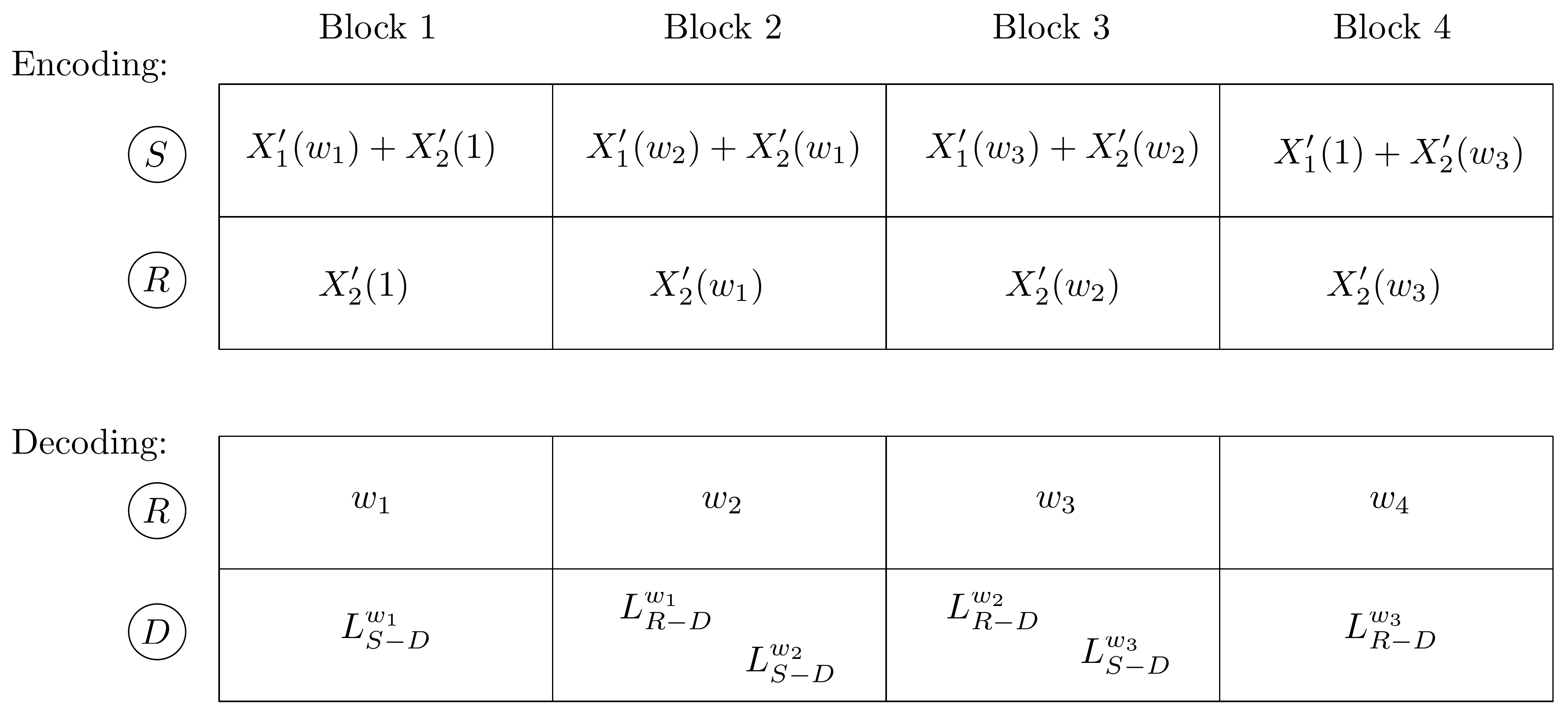}
\caption{Lattice Decode-and-Forward scheme for the AWGN relay channel.} 
\label{DF for onerelay}
\end{figure*}

\noindent{\bf Codebook construction:} 
We consider two nested lattice chains of length three $\Lambda_1 \subseteq \Lambda_{s1}\subseteq \Lambda_{c1}$,  and $\Lambda_2 \subseteq \Lambda_{s2}\subseteq \Lambda_{c2}$ whose existence is guaranteed by Theorem \ref{thm:nam}, and whose parameters $P_i, \gamma$ we still need to specify. 
The nested lattice pairs $(\Lambda_1, \Lambda_{c1})$   and $(\Lambda_2, \Lambda_{c2})$ are used to construct lattice codebooks of coding rate $R$ with $\sigma^2(\Lambda_1) = \alpha P $ and $\sigma^2(\Lambda_2) = \bar{\alpha} P$ for given $\alpha\in [0,1]$. Since $\Lambda_1$ and $\Lambda_2$ will not be the finest lattice in the chain, they will be Rogers good, and hence $\sigma^2(\Lambda_1) = \alpha P $ will define the volume of $\Lambda_1, V_1$, and $\sigma^2(\Lambda_2) = \bar{\alpha} P$ will define the volume of $\Lambda_2, V_2$. 
Since $(\Lambda_1, \Lambda_{c1})$   and $(\Lambda_2, \Lambda_{c2})$ are used to construct lattice codebooks of coding rate 
\[ R = \frac{1}{n}\log\left(\frac{V_1}{V_{c1}}\right) = \frac{1}{n}\log\left(\frac{V_2}{V_{c2}}\right),\]
 this will in turn define $V_{c1}$ in terms of $V_1$ and rate $R$; similarly for $V_{c2}$ in terms of $V_2$ and rate $R$.  
 Since $\Lambda_{c1}$ and $\Lambda_{c2}$ are only Poltyrev good, we may obtain the needed $V_{c1}, V_{c2}$ by appropriate selection of $\gamma$ in Theorem \ref{thm:nam}. Finally, the lattices $\Lambda_{s1}$ and $\Lambda_{s2}$ (whose second moments we may still specify arbitrarily, and which will be used for lattice list decoding at the destination node) will also be Rogers good and their volumes, or equivalently, second moments,  will be selected in the course of the proof.  
Randomly map the messages $w \in \{1,2,\dots,2^{nR}\}$ to codewords ${\bf t_1}(w) \in \mathcal{C}_1 = \{ \Lambda_{c1} \cap \mathcal{V}_1\}$ and ${\bf t_2}(w) \in \mathcal{C}_2 = \{ \Lambda_{c2} \cap \mathcal{V}_2\}$. Let these two mappings be independent and known to all nodes.

We use block Markov coding and define $w_{b}$ as the new message index to be sent in block $b$ ($b=1,2,\cdots, B$); 
 define $w_0 = 1$. At the end of block $b-1$, the receiver knows $(w_1,\dots,w_{b-2} )$ and the relay knows $(w_1,\dots,w_{b-1} )$. 
 We let ${\bf Y_R}(b), {\bf Y_D}(b)$ denote the vectors of length $n$ of received signals at the relay and the destination, respectively, during the $b$-th block, and ${\bf U}_1(b), {\bf U}_2(b)$ denote dithers during block $b$ known to all nodes which are i.i.d.,  change from block to block, and are uniformly distributed over $\mathcal{V}_1$ and $\mathcal{V}_2$ respectively.  The encoding and decoding steps are outlined in Figure \ref{DF for onerelay}. 

\smallskip
\noindent
{\bf Encoding:} During the $b$-th block, the transmitter sends the superposition (sum) ${\bf X_S}(w_b,w_{b-1}) = {\bf X'_{1}}(w_b) +{\bf X'_{2}}(w_{b-1})$,  and the relay sends  ${\bf X_{R}}(w_{b-1})$, where
\begin{align*}
{\bf X'_{1}}(w_b) &= ({\bf t_1}(w_b) - {\bf U_1}(b)) \mod \Lambda_1, \\
{\bf X'_{2}}(w_{b-1}) &= ({\bf t_2}(w_{b-1}) - {\bf U_2}({b-1})) \mod \Lambda_2 \\
{\bf X_{R}}(w_{b-1}) &= \sqrt { \frac{P_R}{\bar{\alpha}P}} {\bf X'_2}(w_{b-1}) = \left(\sqrt { \frac{P_R}{\bar{\alpha}P}} {\bf t_2}(w_{b-1}) - \sqrt { \frac{P_R}{\bar{\alpha}P}}{\bf U_2}({b-1}) \right) \mod \sqrt { \frac{P_R}{\bar{\alpha}P}} \Lambda_2.
\end{align*}
  By the Crypto lemma  ${\bf X'_1}(w_b)$ and ${\bf X'_2}(w_{b-1})$ are uniformly distributed over $\mathcal{V}_1$ and $\mathcal{V}_2$ and independent of all else. 
\smallskip

\noindent
{\bf Decoding:}


1. At the $b$-th block, the relay knows $w_{b-1}$ and consequently ${\bf X'_2}(w_{b-1})$, and so may decode the message $w_b$ from the received signal $ {\bf Y_{R}}(b)  - {\bf X'_2}(w_{b-1}) = {\bf X'_1}(w_b)  + {\bf Z_{R}}(b) $ as long as $ R < C(\alpha P/N_R)$,
since $(\Lambda_1, \Lambda_{c1})$ may achieve the capacity of the point-to-point channel \cite{Erez:2004} or Lemma \ref{lem:unique}.

2. The receiver first decodes a list of messages $w_{b-1}$, $L_{R-D}^{w_{b-1}}({\bf Y_D}(b))$, defined according to \eqref{eq:listdecoding} as 
\begin{equation} L_{R-D}^{w_{b-1}}({\bf Y_D}(b)) = \{w_{b-1} | \; {\bf {t_2}}(w_{b-1}) \in  S_{ \kappa \mathcal{V}_{s2},\kappa \Lambda_{c2}} ({\bf Y_D^\prime}(b)) \mod \kappa \Lambda_2\}, \label{eq:LRD} \end{equation}
 of  asymptotic size $2^{n (R - R_R)}$ from the signal 
 \begin{align}
 {\bf Y_D}(b) &= {\bf X_S}(w_b,w_{b-1})+ {\bf X_R}(w_{b-1}) + {\bf Z_D}(b) \\
 & = {\bf X'_1}(w_b) + \kappa {\bf X'_2}(w_{b-1}) + {\bf Z_D}(b)
 \end{align}
for  $\kappa  =  \left(1+\sqrt{\frac{P_R}{\bar{\alpha}P}}\right)$ 
using the lattice list decoding scheme of Theorem \ref{thm:list}.  
Notice that Theorem \ref{thm:list} is applicable as the ``noise'' in decoding a list of $w_{b-1}$ from ${\bf Y_D}(b)$ is composed of the sum of a Gaussian signal ${\bf Z_D}(b)$ and ${\bf X_1'}(w_b)$ which is uniformly distributed over the fundamental Voronoi region of the Rogers good lattice of second moment $\alpha P$. The equivalent noise variance in Theorem \ref{thm:list} is thus $\alpha P + N_D$, and the capacity of the channel is \cite{Erez:2004} $C(\kappa^2 \bar{\alpha}P / (\alpha P+N_D)) = C((\sqrt{\bar{\alpha P}} + \sqrt{P_R})^2 / (\alpha P+N_D))$. We may thus obtain a list of size $2^{n(R-R_R)}$ as long as 
\begin{align}
R_R &< \frac{1}{2} \log \left( \frac{\kappa^2 \bar{\alpha}P }{ \frac{\kappa^2 \bar{\alpha}P  (\alpha P + N_D) } { \kappa^2 \bar{\alpha}P  + \alpha P + N_D } } \right)
= \frac{1}{2} \log \left( 1 + \frac{ (\sqrt{\bar{\alpha}P} + \sqrt{P_R})^2 }{ \alpha P + N_D } \right). \label{eq:RR}
\end{align}
One may directly apply Theorem \ref{thm:list}; for additional details on this step, please see Appendix \ref{app:list}.

3. A second list of messages $w_{b-1}$ was obtained at the end of block $b-1$ from the direct link  between the transmitter node S and the destination node D, denoted as $L_{S-D}^{w_{b-1}}({\bf Y_D}(b-1) - \kappa{\bf X_2'}(w_{b-2}))$ defined according to \eqref{eq:listdecoding} and analogous to \eqref{eq:LRD} using a lattice list decoder. 
 We now describe the formation of the list $L_{S-D}^{w_b}({\bf Y_D}(b)- \kappa{\bf X_2'}(w_{b-1}))$ in block $b$ which will be used in block $b+1$. Assuming that the receiver has decoded $w_{b-1}$ successfully, it  subtracts $\kappa{\bf X_2'}(w_{b-1})$ from ${\bf Y_D}(b)$: ${\bf Y_D}(b) - \kappa{\bf X_2'}(w_{b-1}) = {\bf X'_1}(w_b) + {\bf Z_D}(b)$, and then decodes another list  of possible messages $w_b$ of  asymptotic size $2^{n(R - C(\alpha P/(N_D)))}$ using Theorem \ref{thm:list}. This is done using the nested lattice chain  $\Lambda_1 \subseteq \Lambda_{s1}\subseteq \Lambda_{c1}$. Again, Theorem \ref{thm:list} is applicable as we have a channel ${\bf X_1'}(w_b)+{\bf Z_D}(b)$ of capacity $C(P/N_D)$ where the noise is purely Gaussian of  second moment $N_D$. Here, choose the list decoding lattice $\Lambda_{s1}$ to have a fundamental Voronoi region of volume approaching $V_{s1} = \left( \frac{N_D}{ \alpha P + N_D} \right)^{n/2} V_1$ {asymptotically (analogous to \eqref{eq:Vs})} so that the size of the decoded list approaches $2^{n(R-C(\alpha P/(N_D)))}$. Notice that this choice of $V_{s1} < V_1$ and hence is permissible by Theorem 2 (as $P_1>P_{s1}$). 
 For the interesting case when $R$ approaches $\frac{1}{2}\log\left(1+\frac{P+P_R+2\sqrt{\bar{\alpha}P P_R}}{N_D} \right) $ (and hence list decoding is needed / relevant),  $V_{c1} = \left( \frac{N_D}{ P + P_R + 2\sqrt{\bar{\alpha}PP_R} + N_D} \right)^{n/2} V_1$ asymptotically in the sense of \eqref{eq:Vs}. Thus $V_{c1} < V_{s1} < V_1$ as needed.

 4. The receiver now 
 decodes $w_{b-1}$ by  intersecting two independent lists $L_{R-D}^{w_{b-1}}({\bf Y_D}(b))$ and $L_{S-D}^{w_{b-1}}({\bf Y_D}(b-1)- \kappa{\bf X_2'}(w_{b-2}))$ and declares a success if there is a unique $w_{b-1}$  in this intersection. 
 Errors are declared if there is no, or multiple messages in this intersection. We are guaranteed by Theorem \ref{thm:list} that the correct message will lie in each list, and hence also in their intersection, with high probability by appropriate choice $V_{s1}$ and $V_{s2}$.  To see that no more than one message will lie in the list, notice that the two lists are {\it independent} due to the random and independent mappings between the message and two codeword sets.  Thus, following the arguments surrounding \cite[Eq. (27) and Lemma 3]{Cover:1979}, or alternatively by independence of the lists and applying \cite[Packing Lemma]{elgamalkim}, with high probability, there is no more than one correct message in this intersection if $R - C(\alpha P/(N_2)) -R_R <0$, or
\begin{align*} R & < \frac{1}{2} \log \left( 1 + \frac{\alpha P}{N_D} \right) + R_R < \frac{1}{2} \log \left( 1 + \frac{ P + P_R + 2\sqrt{\bar{\alpha} P P_R} } { N_D} \right).
\end{align*}

\end{proof}

{\begin{remark}
While we have mimicked the regular encoding / sliding window decoding method to achieve the DF rate, lattice list decoding may equally be used in the irregular encoding and backwards decoding schemes. 
The intuition we want to reinforce is that one may obtain similar results to random-coding based DF schemes using lattice codes by intersecting multiple independent lists to decode a unique message. Furthermore, as the lattice list decoder is a Euclidean lattice decoder, it does not increase the complexity at the decoder. We note that using lists is not necessary -- other novel lattice-based schemes can be used instead of lattice list decoding such as \cite{Nokleby:2011} to achieve the same DF rate region.
\end{remark}}

\subsection{DF for the multi-relay Gaussian relay channel}
We now show that nested lattice codes may also be used to achieve the DF rates of the single source, single destination multi-level relay channel \cite{xie_kumar_04, xie2005achievable, Kramer:Gastpar:Gupta}. Here, all definitions remain the same as in Section \ref{subsec:DF}; changing the channel model to account for an arbitrary number of full-duplex relays. For the 2 relay scenario we show the input/output relations used in deriving achievable rates in Figure \ref{fig:relay-channel}. In general we would for example have $Y_2 = X_1+X_2+X_3+Z_2$, but that, for our achievability scheme we assume a relay order (e.g. 2 then 3) which results in the equivalent input/output equation $Y_2=X_1+Z_2$ at node 2. This is equivalent due to the achievability scheme we will propose combined with the assumed relaying order, in which node 2 will be able to cancel out all signals transmitted by itself as well as node 3 (more generally, node $i$ may cancel  out all relay transmissions ``further'' in the relay order than itself).  

The central idea remains the same -- we cooperate via a series of lattice list decoders and replace multiple joint typicality checks with the intersection of multiple independent lists obtained via the lattice list decoder. For clarity, we focus on the two-relay case as in Figure \ref{fig:relay-channel}, but the results may be extended to the $N$-relay case in a straightforward manner.  Let $\pi(\cdot)$ denote a permutation (or ordering) of the relays. In the $N=2$ case as shown in Figure \ref{fig:relay-channel} we have two possible permutations:  the first the identity permutation $\pi(2) = 2, \pi(3)=3$ and  the second $\pi(2) = 3, \pi(3)=2$. 


The channel model is expressed as (a node's own signal is omitted as it may subtract it off) 
\begin{align*}
{\bf Y_2} &= {\bf X_1} + {\bf X_3} + {\bf Z_2} \\
{\bf Y_3} &= {\bf X_1} + {\bf X_2} + {\bf Z_3} \\
{\bf Y_4} &= {\bf X_1} + {\bf X_2} + {\bf X_3} + {\bf Z_4},
\end{align*}
where ${\bf Z_2} \sim{\cal N}({\bf 0}, N_2{\bf I})$, ${\bf Z_3}\sim{\cal N}({\bf 0}, N_3{\bf I})$ and ${\bf Z_4}\sim{\cal N}({\bf 0}, N_4{\bf I})$, and nodes are subject to input power constraints
$\frac{1}{n} E[{\bf X_1}^T{\bf X_1}] \leq P_1$ , $\frac{1}{n} E[{\bf X_2}^T{\bf X_2}]\leq P_2$, and $\frac{1}{n} E[{\bf X_3}^T{\bf X_3}]\leq P_3$.

\begin{theorem}
{\it Lattices achieve the DF rate achieved by Gaussian random codebooks for the multi-relay channel.} The following rate $R$ is achievable using nested lattice codes for the Gaussian two relay channel described by \cite{xie_kumar_04}:
\begin{align*}
R < & \max_{\pi(\cdot)} \max_{0 \leq \alpha_1, \beta_1, \alpha_2 \leq 1} \min  \left\{ C \left(  \frac{\alpha_1 P_1}{N_{\pi(2)}} \right) , C \left( \frac{\alpha_1 P_1 + (\sqrt{\beta_1P_1} + \sqrt{\alpha_2 P_{\pi(2)}})^2}{N_{\pi(3)}} \right) \right., \\
& \left. C\left(\frac{\alpha_1P_1 + \left(\sqrt{\beta_1P_1} + \sqrt{\alpha_2P_{\pi(2)}}\right)^2 + \left(\sqrt{(1 - \alpha_1-\beta_1P_1)} + \sqrt{(1-\alpha_2)P_{\pi(2)}} +\sqrt{P_{\pi(3)}} \right)^2 }{ N_4}\right) \right\} 
\end{align*}
\label{thm:DFm}
\end{theorem}
The proof of Theorem \ref{thm:DFm} may be found in Appendix \ref{app:DFm}, and follows along the same lines as Theorem \ref{thm:DF}.
\section{Multi-source Decode and Forward -- combining compute-and-forward and DF}
\label{sec:DFm}

We now illustrate how  list decoding may be combined with the linearity of lattice codes in more general networks by considering two examples. In particular, we consider relay networks in which two messages are communicated, along relayed and direct links, as opposed to the single message case previously considered. The relay channel may be viewed as strictly cooperative in the sense that all nodes aid in the transmission of the same message and the only impairment is noise; the presence of multiple messages leads to the notion of interference and the possibility of decoding combinations of messages. 

We again focus on demonstrating the utility of lattices in DF-based achievability schemes. In the previous section it was demonstrated that lattices may achieve {\it the same} rates as Gaussian random coding based schemes. Here, the presence of multiple messages/sources gives lattices a potential rate benefit over random coding-based schemes, as encoders and decoders may exploit the linearity of the lattice codes to better decode a linear combination of messages. Often, such a linear combination is sufficient to extract the desired messages if combined with the appropriate side-information, and may enlarge the achievable rate region for certain channel conditions.   
In this section, we demonstrate two examples of combining Compute-and-Forward based decoding of the sum of signals at relays with direct link side-information in:
1) the two-way relay channel with direct links and 2) the multiple-access relay channel. To the best of our knowledge, these are the first lattice-coding based achievable rate regions for these channels. 

\begin{figure}
\centering
\includegraphics[width=16cm]{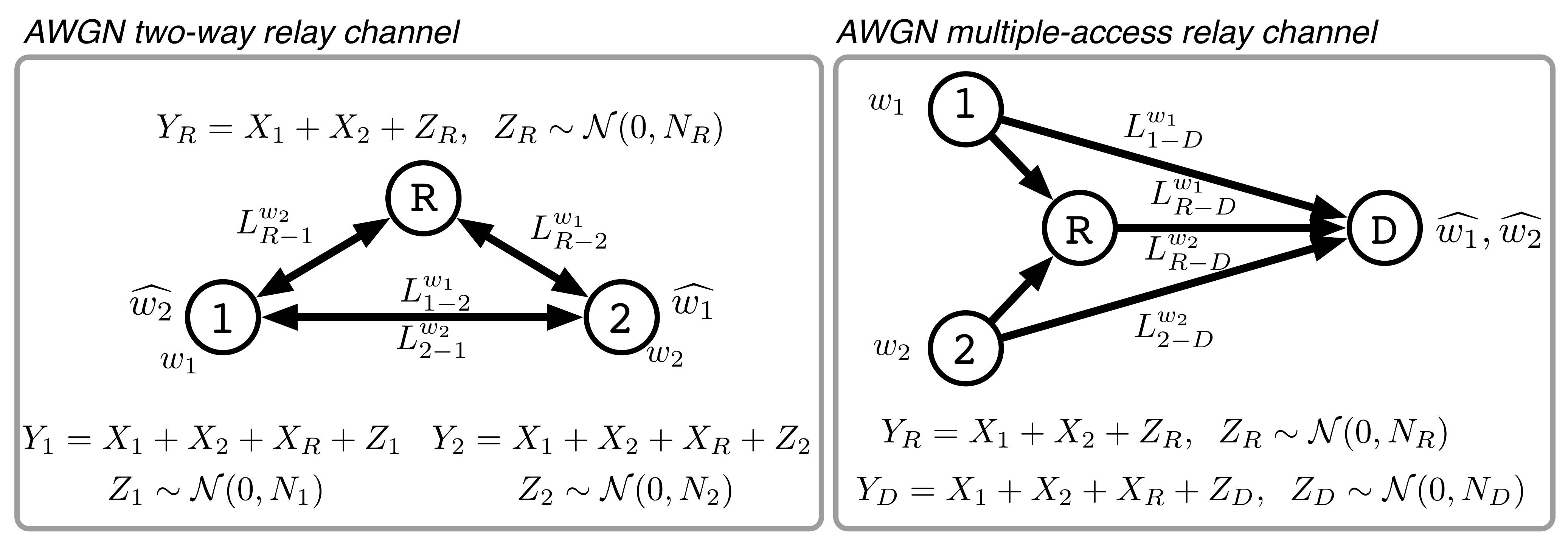}
\caption{The AWGN two-way relay channel with direct links and the AWGN multiple-access relay channel. We illustrate the lists $L_{i-j}^{w}$ of messages $w$ carried by the codewords at node $i$ and  list decoded according to Theorem \ref{thm:list} at node $j$.}
\label{fig:multi-source relay}
\end{figure}

\subsection{The two-way Gaussian relay channel with direct links}
\label{subsec:TWRC}

The two-way relay channel is the logical extension of the classical relay channel for one-way point-to-point  communication aided by a relay to allow for two-way communication.
While the capacity region is in general unknown, it is known for half-duplex channel models under the 2-phase MABC protocol  \cite{Kim:comparison}, to within 1/2 bit for the full-duplex Gaussian channel model with no direct links  \cite{Nam:IEEE,Narayanan:2010}, and to within 2 bits for the same model with direct links in certain cases \cite{avestimehr2010capacity}.


Random coding techniques employing  DF, CF, and AF relays have been the most common in deriving achievable rate regions for the two-way relay channel, but a handful of work \cite{Baik:2008, Nam:IEEE, Narayanan:2010, ong2010capacity} has considered lattice-based schemes which, in a DF-like setting,  effectively exploit the additive nature of the Gaussian noise channel in allowing the sum  of the two transmitted lattice points to be decoded at the relay. The intuitive gains of decoding the sum of the messages rather than the individual messages stem from the absence of the classical multiple-access sum constraints.  This sum-rate point is forwarded to the terminal 
which utilizes its own-message side-information to subtract off its own message from the decoded sum.
While random coding schemes have been used in deriving achievable rate regions in the presence of direct links, lattice codes -- of interest in order to exploit the ability to decode the sum of messages at the relay -- have so far not been used. 
We present such a lattice-based scheme next. 


The two-way Gaussian relay channel with direct links consists of two terminal nodes with inputs $X_1, X_2$ with power constraints $P_1, P_2$ (without loss of generality, it is assumed $P_1 \geq P_2$) and outputs $Y_1, Y_2$ which wish to exchange messages $w_1\in \{1,2,\cdots , 2^{nR_1}\}$ and $w_2\in \{1,2,\cdots, 2^{nR_2}\}$ with the help of the relay with input $X_R$ of power $P_R$ and output $Y_R$. We assume, without loss of generality (WLOG),  the channel:
\begin{align*}
{\bf Y_1} &= {\bf X_R} + h_{21}{\bf X_2} + {\bf Z_1}, \;\;\;\; {\bf Z_1} \sim {\cal N}({\bf 0},N_1{\bf I}) \\
{\bf Y_2} &= {\bf X_R} + h_{12}{\bf X_1} + {\bf Z_2}, \;\;\;\; {\bf Z_2} \sim {\cal N}({\bf 0},N_2{\bf I}) \\
{\bf Y_R} &= {\bf X_1} + {\bf X_2} + {\bf Z_R}, \;\;\;\; {\bf Z_R} \sim {\cal N}({\bf 0},N_R{\bf I}),
\end{align*}
subject to input power constraints $\frac{1}{n} E[{\bf X_1}^T{\bf X_1}] \leq P_1, \, \frac{1}{n} E[{\bf X_2}^T{\bf X_2}] \leq P_2, \, \frac{1}{n} E[{\bf X_R}^T{\bf X_R}] \leq P_R$ and real constants $h_{12}, h_{21}$. 
The channel model is shown in Figure \ref{fig:multi-source relay}, and all input and output alphabets are  ${\mathbb R}$.

An $(2^{nR_1}, 2^{nR_2}, n)$ code for the two-relay channel consists of the two sets of messages $w_i$, $i=1,2$ uniformly distributed 
over  ${\cal M}_i : = \{1,2,\cdots , 2^{nR_i}\}$,  and two encoding functions $X_i^n: {\cal M}_i \rightarrow {\mathbb R}^n$ (shortened to ${\bf X_i}$) satisfying the 
power constraints $P_i$, a set of relay functions $\{f_j\}_{j=1}^n$ such that the relay channel input at time $j$ is a function of the 
previously received relay channel outputs from channel uses $1$ to $j-1$, $X_{R,j} = f_j(Y_{R,1}, \cdots,Y_{R,j-1})$, and finally two 
decoding functions $g_i: {\cal Y}_i^n \times {\cal M}_i \rightarrow {\cal M}_{\bar{i}}$ which yields the message estimates $\hat{w}_{\bar{i}}: = g_i(Y_i^n, w_i)$ for $\bar{i} = \{1,2\}\setminus i$. 
We define the average probability of error of the code to be $P_{n,e} : = \frac{1}{2^{n(R_1+R_2)}} \sum_{w_1\in {\cal M}_1, w_2\in {\cal M}_2} \Pr\{(\hat{w_1}, \hat{w_2}) \neq (w_1,w_2)|(w_1,w_2) \mbox{ sent}\}$. The rate 
pair $(R_1,R_2)$ is then said to be achievable by the two-relay channel if, for any $\epsilon>0$ and for sufficiently large $n$, there 
exists an $(2^{nR_1},2^{nR_2},n)$ code such that $P_{n,e} < \epsilon$. The capacity region $C$ of the two-way relay channel is the 
supremum of the set of achievable rate pairs.

\begin{theorem}
\label{thm:two-way}
{\it Lattices in two-way relay channels with direct links.} The following rates are achievable for the two-way AWGN relay channel with direct links
\begin{align}
R_1 & \leq \min \left( \left[\frac{1}{2} \log \left( \frac{P_1}{P_1 + P_2} + \frac{P_1}{N_R}\right)\right]^+ , \frac{1}{2} \log\left(1+\frac{h_{12}^2P_1+P_R}{N_2}\right)\right) \label{eq:R21}\\
R_2 &\leq  \min \left( \left[\frac{1}{2} \log \left( \frac{P_2}{P_1 + P_2} + \frac{P_2}{N_R}\right)\right]^+ , \frac{1}{2} \log\left(1+\frac{h_{21}^2P_2+P_R}{N_1}\right)\right) \label{eq:R22}.
\end{align}
\end{theorem}

\begin{proof}
The achievability proof combines a lattice version of regular encoding/sliding window decoding scheme (to take advantage of the direct link), decoding of the sum of transmitted signals at the relay using nested coarse lattices to take care of the asymmetric powers, as in \cite{Nam:IEEE}, a lattice  binning technique equivalent to the random binning technique developed by \cite{xie2007network}, and lattice list decoding at the terminal nodes to combine direct and relayed information.

\begin{figure*}
\centering
\includegraphics[width=14cm]{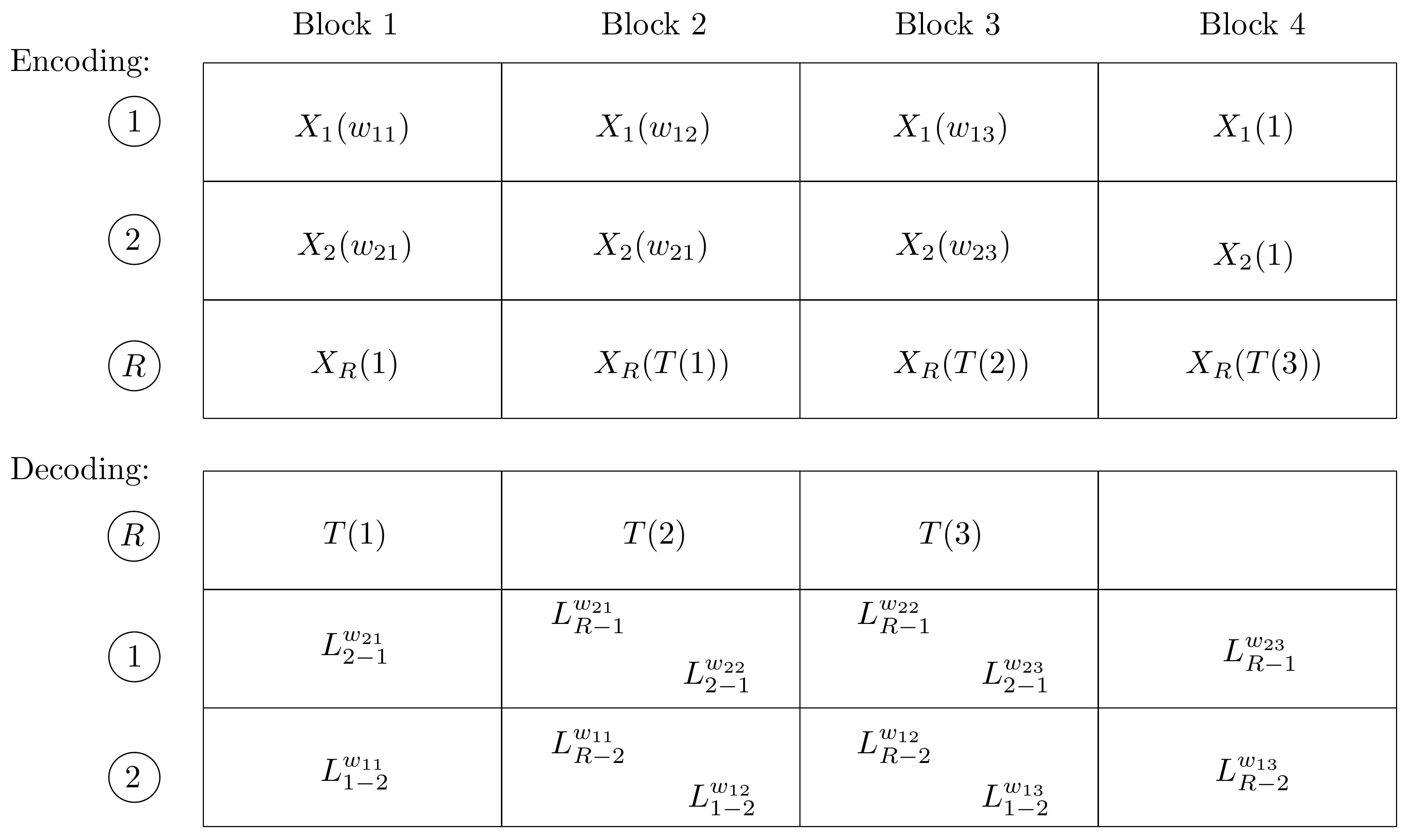}
\caption{Lattice Decode-and-Forward scheme for the AWGN two-way relay channel with direct links. } 
\label{DF for twoway}
\end{figure*}

{\bf Codebook construction:} 
Assume WLOG that $P_1>P_2$.  We construct two nested lattice chains according to Theorem \ref{thm:nam}. 
The first consists of the lattices $\Lambda_1, \Lambda_2, \Lambda_{s1}, \Lambda_{s2}, \Lambda_{c1}, \Lambda_{c2}$ all nested such that:
\begin{itemize}
\item  $\Lambda_1 \subseteq \Lambda_{s1}\subseteq \Lambda_{c1}$  and  $\Lambda_2 \subseteq \Lambda_{s2}\subseteq \Lambda_{c2}$.; the coarsest lattice is $\Lambda_1$ or $\Lambda_2$ and the finest is $\Lambda_{c1}$ or $\Lambda_{c2}$. 
\item $\sigma^2(\Lambda_1) = P_1, \sigma^2(\Lambda_2)=P_2$
\item the coding rate of $(\Lambda_1, \Lambda_{c1})$ is $R_1 = \frac{1}{n}\log\left(\frac{V_1}{V_{c1}}\right)= \frac{1}{2}\log\left(\frac{P_1}{P_{c1}}\right) + o_n(1)$, and that of $(\Lambda_2, \Lambda_{c2})$ is $R_2 = \frac{1}{n}\log\left(\frac{V_2}{V_{c2}}\right)= \frac{1}{2}\log\left(\frac{P_2}{P_{c2}}\right) + o_n(1)$. 
Associate each message $w_1 \in \{1,\dots,2^{nR_1}\}$ with ${\bf t_1}(w_1) \in \mathcal{C}_1 = \{ \Lambda_{c1} \cap \mathcal{V}_1 \}$ and each message $w_2 \in \{1,\dots,2^{nR_2}\}$ with ${\bf t_2}(w_2) \in \mathcal{C}_2 = \{ \Lambda_{c2} \cap \mathcal{V}_2 \}$. 
\item if $V_{c1}>V_{c2}$ (determined by relative values of $R_1,P_1$ and $R_2, P_2$ in the above), then $\Lambda_{c1}\subseteq \Lambda_{c2}$, implying $\Lambda_{c1}$ may be Rogers good and hence we may  guarantee the desired $V_{c1}$ by proper selection of $P_{c1}$ in Theorem \ref{thm:nam} $\left( \mbox{as }R_1  = \frac{1}{2}\log\left(\frac{P_1}{P_{c1}}\right) + o_n(1) = \frac{1}{n}\log\left(\frac{V_1}{V_{c1}}\right)\right)$; otherwise by proper selection of $\gamma$ 
in Theorem \ref{thm:nam} (and likewise for $\Lambda_{c2}$). 
\item the lattices $\Lambda_{s1}$ and $\Lambda_{s2}$ which will be 
used for lattice list decoding at node 2 and 1 respectively are both Rogers good and hence may be specified by the volumes of their 
fundamental Voronoi  regions $V_{s1}$ and $V_{s2}$ (under the constraints $V_1\geq V_{s1}\geq V_{c1}$ and $V_2\geq V_{s2}\geq V_{c2}$), or the corresponding $P_{c1}, P_{c2}$. These will be chosen in the course of the proof. 
\item Then final relative ordering of the six lattices 
will then depend on the relative sizes of their fundamental region volumes.
\end{itemize}

We also construct a nested lattice chain of $\Lambda_R, \Lambda_{sR1}, \Lambda_{sR2}, \Lambda_{cR}$ according to Theorem \ref{thm:nam} such that:
\begin{itemize}
\item  $\Lambda_R \subseteq \Lambda_{sR1} \subseteq \Lambda_{sR2} \subseteq \Lambda_{cR}$ or  $\Lambda_R \subseteq \Lambda_{sR2} \subseteq \Lambda_{sR1} \subseteq \Lambda_{cR}$
\item $\sigma^2(\Lambda_R) = P_R$
\item the relay uses the codebook $\mathcal{C}_R = \{ \Lambda_{cR} \cap \mathcal{V}_R \}$ consisting of codewords ${\bf t_R}$. This codebook is of rate $R_R = \frac{1}{n}\log\left(\frac{V_R}{V_{cR}}\right) = \frac{1}{n}\log\left(\frac{V_1}{V_{c1}}\right)$ if $\Lambda_{c2}\subseteq \Lambda_{c1}$ and of rate  $R_R = \frac{1}{n}\log\left(\frac{V_R}{V_{cR}}\right) = \frac{1}{n}\log\left(\frac{V_1}{V_{c2}}\right)$ if $\Lambda_{c1}\subseteq \Lambda_{c2}$. This rate $R_R$ in turn fixes the choice of $\gamma$ in Theorem \ref{thm:nam}.  
\item $\Lambda_{sR1}$ and $\Lambda_{sR2}$ are used to decode lists at the two destinations, and their 
relative nesting depends on $V_{sR1}$ and $V_{sR2}$ (or equivalently $P_{sR1}$ and $P_{sR2}$ as both are Rogers good) subject to $V_R\geq V_{sR1}\geq V_{cR}$ and $V_{cR}\geq V_{sR2}\geq V_R$ which will be specified in the course of the proof. 
\end{itemize}

\smallskip
\noindent
{\bf Encoding:} We use Block Markov encoding. Messages $w_{1b}\in \{1,2\cdots 2^{nR_1}\}$ and $w_{2b}\in \{1,2,\cdots 2^{nR_2}\}$ are the messages the two terminals wish to send in block $b$. 
Nodes 1 and 2 send ${\bf X_1}(w_{1b})$ and ${\bf X_2}(w_{2b})$:
\begin{align*}
 {\bf X_1}(w_{1b}) &= ({\bf t_1}(w_{1b}) - {\bf U_1}({b})) \mod \Lambda_1 \\
 {\bf X_2}(w_{2b}) &= ({\bf t_2}(w_{2b}) - {\bf U_2}(b)) \mod \Lambda_2, \end{align*}
 for dithers ${\bf U_1}(b), {\bf U_2}(b)$ known to all nodes which are i.i.d. uniformly distributed over ${\cal V}_1$ and ${\cal V}_2$ and vary from block to block. At the relay,  we assume that it has obtained
 \begin{equation} {\bf T}(b-1) = ({\bf t_1}(w_{1(b-1)}) + {\bf t_2}(w_{2(b-1)}) - Q_2({\bf t_2}(w_{2(b-1)}) + {\bf U_2}(b-1)) ) \mod \Lambda_1 \label{eq:T1}
\end{equation}
  in block $b-1$. 
Note that ${\bf T}(b-1)$ lies in $\{\Lambda_{c2}\cap {\cal V}_1\}$ if $\Lambda_{c1}\subseteq \Lambda_{c2}$ and in $\{\Lambda_{c1}\cap {\cal V}_1\}$ if $\Lambda_{c2}\subseteq \Lambda_{c1}$, and is furthermore uniformly distributed over this set consisting of $2^{nR_R}$ points.  We may thus associate each ${\bf T}(b-1)$ with an index say $i({\bf T}(b-1))$, 
  which the relay then uses as index for the codeword ${\bf t_R}(i({\bf T}(b-1)))$ in ${\cal C}_R$ (also of rate $R_R$). 
  To simplify notation and with some abuse of notation we write ${\bf  t_R} ({\bf T }(b-1) )$ instead of the indexed version ${\bf t_R}(i({\bf T}(b-1)))$.  
   The relay then sends 
\begin{equation}
{\bf X_R}({\bf T}(b-1)) = ( {\bf  t_R} (  {\bf T }(b-1) ) + {\bf U_R}(b-1) ) \mod \Lambda_R,\label{eq:T2}
\end{equation}
 for ${\bf U_R}(b-1)$ a dither known to all nodes which is uniformly distributed over ${\cal V}_R$.

\smallskip
\noindent
{\bf Decoding:} During block $b$, the following messages / signals are known / decoded at each node:
\begin{itemize}
\item Node 1:  knows $w_{11}, \cdots w_{1b}$, $w_{21}, w_{22}, \cdots w_{2(b-2)}$, decodes $w_{2(b-1)}$ 
\item Node 2:  knows $w_{21}, \cdots w_{2b}$, $w_{11}, w_{12}, \cdots w_{1(b-2)}$, decodes $w_{1(b-1)}$ 
\item Node $R$: knows ${\bf T}(1), {\bf T}(2), \cdots {\bf T}(b-1)$, decodes ${\bf T}(b)$
\end{itemize}
{\it Relay decoding:} The relay terminal receives $ {\bf Y_R}(b) = {\bf X_1} (w_{1b}) + {\bf X_2} (w_{2b}) + {\bf Z_R}(b)$, and, following the arguments of \cite{nazer2011compute, Nam:IEEE, Narayanan:2010}  can decode ${\bf T}(b) = ({\bf t_1}(w_{1b}) + {\bf t_2}(w_{2b}) - Q_2({\bf t_2}(w_{2b}) + {\bf U_2}(b))) \mod \Lambda_1 $ if
 \begin{align*}
R_1  \leq \left[\frac{1}{2} \log \left( \frac{P_1}{P_1 + P_2} + \frac{P_1}{N_R}\right)\right]^+,  \;\;   R_2 \leq \left[\frac{1}{2} \log \left( \frac{P_2}{P_1 + P_2} + \frac{P_2}{N_R}\right)\right]^+.
\end{align*}


{\it Terminal 2 decoding:} Terminal 2 decodes $w_{1(b-1)}$ after block $b$ from the received signals 
\begin{align*}
 {\bf Y_2}(b-1) &= {\bf X_R} ( {\bf T}(b-2) )  + h_{12}{\bf X_1} (w_{1(b-1)}) + {\bf Z_2}(b-1) \\
 {\bf Y_2}(b) &= {\bf X_R} ( {\bf T}(b-1) )  + h_{12}{\bf X_1} (w_{1b}) + {\bf Z_2}(b).
 \end{align*}
This will generally follow the lattice version of regular encoding/sliding-window decoding scheme as described in Section \ref{subsec:DF}. That is, after block $b-1$, terminal 2 first forms ${\bf Y_2^*}(b-1) = {\bf Y_2}(b-1) -  {\bf X_R} ( {\bf T}(b-2) )$ since it has decoded $w_{1(b-2)}$ and knows its own $w_{2(b-2)}$ and hence may form  ${\bf X_R} ( {\bf T}(b-2) )$. 
Then it uses the list decoder of Theorem \ref{thm:list} to produce a list of messages
$w_{1(b-1)}$,
denoted by $L_{1-2}^{w_{1(b-1)}}({\bf Y_2^*}(b-1))$, of size $2^{n(R_1 - C(h_{12}^2 P_1/N_2))}$ using the lattice $\Lambda_{s1}$, whose fundamental Voronoi  region volume is selected to asymptotically approach $V_{s1} = \left(\frac{N_2}{h_{12}^2P_1+N_2} \right)^{n/2}V_1$ (in the sense of \eqref{eq:Vs}).  For $R$ approaching $\frac{1}{2} \log\left(1+\frac{h_{12}^2P_1+P_R}{N_2}\right)$, where list decoding is relevant, $V_{c1} = \left(\frac{N_2}{h_{12}^2P_1+ P_R +N_2} \right)^{n/2}V_1$ asymptotically, and thus $V_{c1} < V_{s1}<V_1$ as needed. To resolve which codeword was actually sent, it intersects this list with another list $L_{R-2}^{w_{1(b-1)}}({\bf Y_2}(b))$ of $w_{1(b-1)}$ obtained in this block $b$. 
This list $L_{R-2}^{w_{1(b-1)}}({\bf Y_2}(b))$ of messages $w_{1(b-1)}$ is obtained from ${\bf Y_2}(b)$ using lattice list decoding with the  lattice $\Lambda_{sR2}$ whose fundamental Voronoi  region volume is taken to asymptotically approach $V_{sR2} = \left( \frac{h_{12}^2P_1+N_2}{P_R+h_{12}^2P_1+N_2}\right)^{n/2}V_R$. For $R$ approaching $\frac{1}{2} \log\left(1+\frac{h_{12}^2P_1+P_R}{N_2}\right)$, where list decoding is relevant,  $V_{cR} = \left(\frac{N_2}{h_{12}^2P_1+ P_R +N_2} \right)^{n/2}V_R$ asymptotically, and thus $V_{cR} < V_{sR2}<V_R$ as needed.  One may verify that by construction of the nested lattice chains, all conditions of Theorem \ref{thm:list} are met. 
This list of messages $w_{1(b-1)}$ is actually obtained from decoding a list of ${\bf t_R}({\bf T}(b-1))$, and using knowledge of its own 
${\bf t_2}(w_{2(b-1)})$ to obtain a list of ${\bf t_1}(w_{1(b-1)})$ (and hence $w_{1(b-1)}$ by one-to-one mapping) of size approximately $2^{n(R_1 - C(\frac{P_R}{h_{12}^2P_1 + N_2}) ) }$. To see this, notice that each ${\bf t}_R$ is associated with a single ${\bf T}= ({\bf t_1} + {\bf t_2}- Q_2({\bf t_2}+ {\bf U_2}) \mod \Lambda_1$. Then, given ${\bf T}$ and ${\bf t_2}$, one may obtain a single ${\bf t_1}$ as follows:
{\begin{align}
(&{\bf T} - {\bf t_2} + Q_2({\bf t_2} + {\bf U_2}) )\mod \Lambda_1 \nonumber \\
&= ( ({\bf t_1} + {\bf t_2} - Q_2({\bf t_2} + {\bf U_2})) - {\bf t_2} + Q_2({\bf t_2} + {\bf U_2}) ) \mod \Lambda_1 \nonumber  \\
& = {\bf t_1} \mod \Lambda_1 ={\bf t_1}. \label{eq:equiv1}
\end{align}}
Similarly, given a ${\bf T}$ and ${\bf t_1}$ one may obtain a single ${\bf t_2}$ as 
{ \begin{align}
(&{\bf T} \mod \Lambda_2 - {\bf t_1}) \mod \Lambda_2 \\
&= ( ({\bf t_1} + {\bf t_2} - Q_2({\bf t_2} + {\bf U_2}) ) \mod \Lambda_1 \mod \Lambda_2 - {\bf t_1} )\mod \Lambda_2 \nonumber  \\
&\overset{(a)}{=}  ( ({\bf t_1} + {\bf t_2} - Q_2({\bf t_2} + {\bf U_2}) ) \mod \Lambda_2 - {\bf t_1} )\mod \Lambda_2 \nonumber \\
&= {\bf t_2} \mod \Lambda_2 = {\bf t_2}, \label{eq:equiv2}
\end{align}}
where $(a)$ follows from ${\bf X}\mod \Lambda_1 \mod \Lambda_2 = {\bf X} \mod \Lambda_2 $ when $\Lambda_1\subseteq \Lambda_2$.
Hence, the list of decoded codewords ${\bf t_R}$ may be transformed into a list of ${\bf t_1}$ at Terminal node 2, which may in turn be associated with a list of $w_{1(b-1)}$. 
The two decoded lists of  $w_{1(b-1)}$ are independent due to the independent mapping relationships between $w_1$ and ${\bf t_1}$  at Node 1 and between ${\bf T}$ and ${\bf t_R}$ at the relay. 
List decoding ensures that at least the correct message lies in the intersection with high probability. To ensure no more than one in the intersection, 
 \begin{align*}
 R_1 &< C(P_R/(h_{12}^2P_1 + N_2) ) + C(h_{12}^2P_1/N_2)\\
 &= C((h_{12}^2P_1 + P_R) /N_2).
 \end{align*}


Analogous steps apply to rate $R_2$.
\end{proof}

\subsection{Comparison to existing rate regions}
We briefly compare the new achievable rate region of Theorem \ref{thm:two-way} with three other existing Decode-and-Forward based rate regions for the two-way relay channel with direct links, and to the cut-set outer bound. In particular, in Figure \ref{fig:numerical}, the region ``Rankov-DF'' \cite[Proposition 2]{Rankov:2006}, the blue ``Xie'' \cite[Theorem 3.1 under Gaussian inputs]{xie2007network} and our orange ``This work'' (Theorem \ref{thm:two-way}) are compared to the green cut-set outer bound under three different choices of noise and power constraints for $h_{12}=h_{21}=1$. 
The ``Rankov-DF'' and ``Xie'' schemes use a multiple access channel model to decode the two messages at the relay, while we use lattice codes to decode their sum, which avoids the sum rate constraint. In the broadcast phase, the ``Rankov-DF'' scheme broadcasts the superposition of the two codewords, while the ``Xie'' and our scheme use a random binning technique to broadcast the bin index. The advantage of  the ``Rankov-DF'' scheme is its ability of obtain a coherent gain at the receiver from the source and 
relay at the cost of a reduced power for each message (power split $\alpha P$ and $(1-\alpha)P$). On the other hand, the ``Xie'' and Theorem \ref{thm:two-way} schemes both broadcast the bin index using all of the relay power, but are unable to obtain coherent gains.  We note that our current scheme does not allow for a coherent gain between the direct and relayed links as 1) we decode the sum of codewords and re-encode that, and 2) we use the full relay power to transmit this sum. Whether simultaneous coherent gains are possible to the two receivers while using a lattice-based scheme to decode the sum of codewords is an interesting open question which may possibly be addressed along the lines of \cite{Nokleby:ISWCS}.

At low SNR, the rate-gain seen by decoding the sum and eliminating the sum-rate constraint is outweighed by 1) the loss seen in the rates $\frac{1}{2} \log \left( \frac{P_i}{P_1 + P_2} + SNR\right)$ compared to $\frac{1}{2} \log ( 1 + SNR)$, or 2) the coherent gain present in the ``Rankov-DF'' scheme. At high SNR, our scheme performs well, and at least in some cases, is able to guarantee an improved finite-gap result to the outer bound, as further elaborated upon in \cite{song:allerton:2010}.
Further note that, compared with the two-way relay channel without direct links \cite{Narayanan:2010, Nam:IEEE}, the direct links may provide additional information which translate to rate gains -- direct comparison shows that the rate region in \cite[Theorem 1]{Nam:IEEE} is always contained in that of Theorem \ref{thm:two-way}.

\begin{figure*}
\subfigure{\includegraphics[width=1.9in]{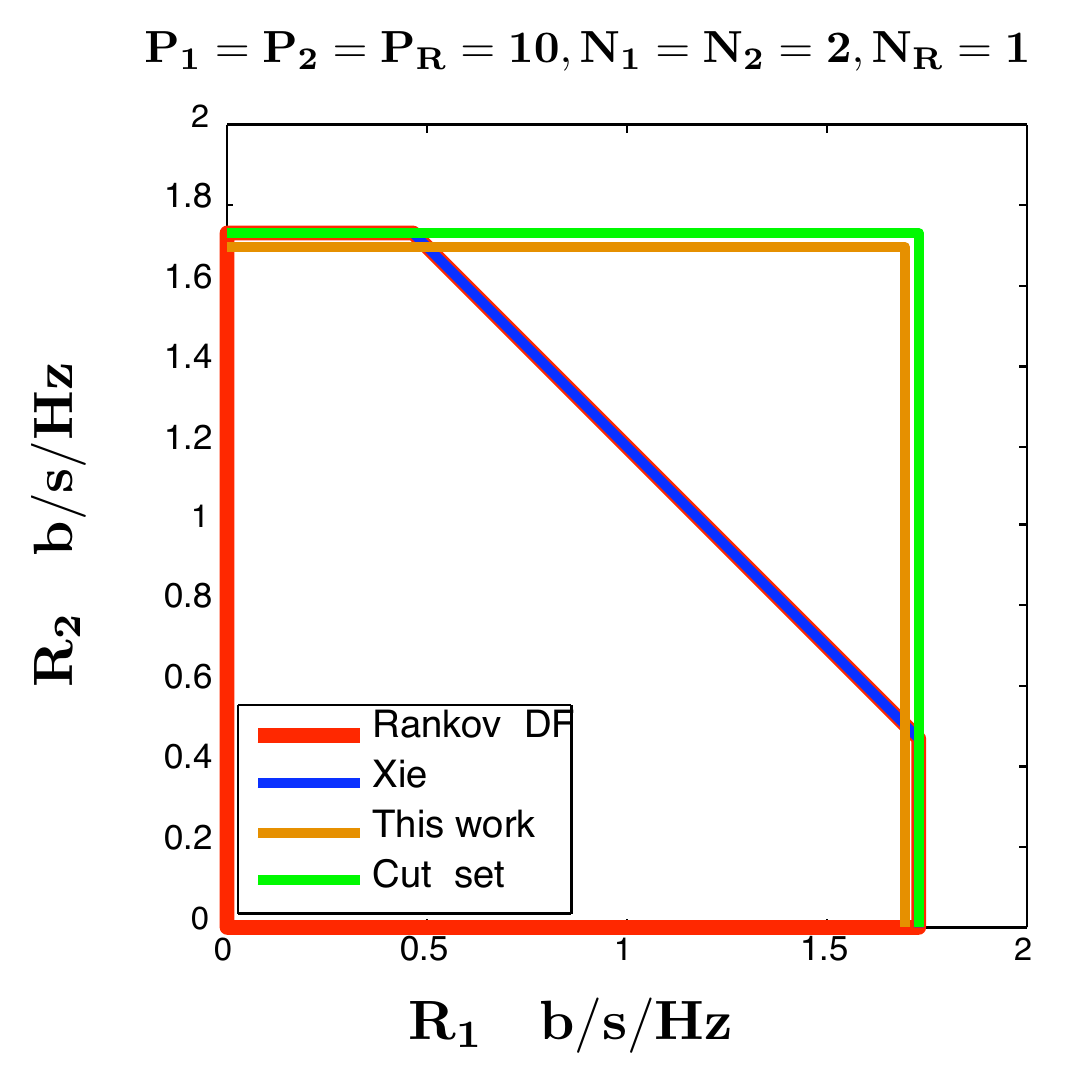}}
\subfigure{\includegraphics[width=2in]{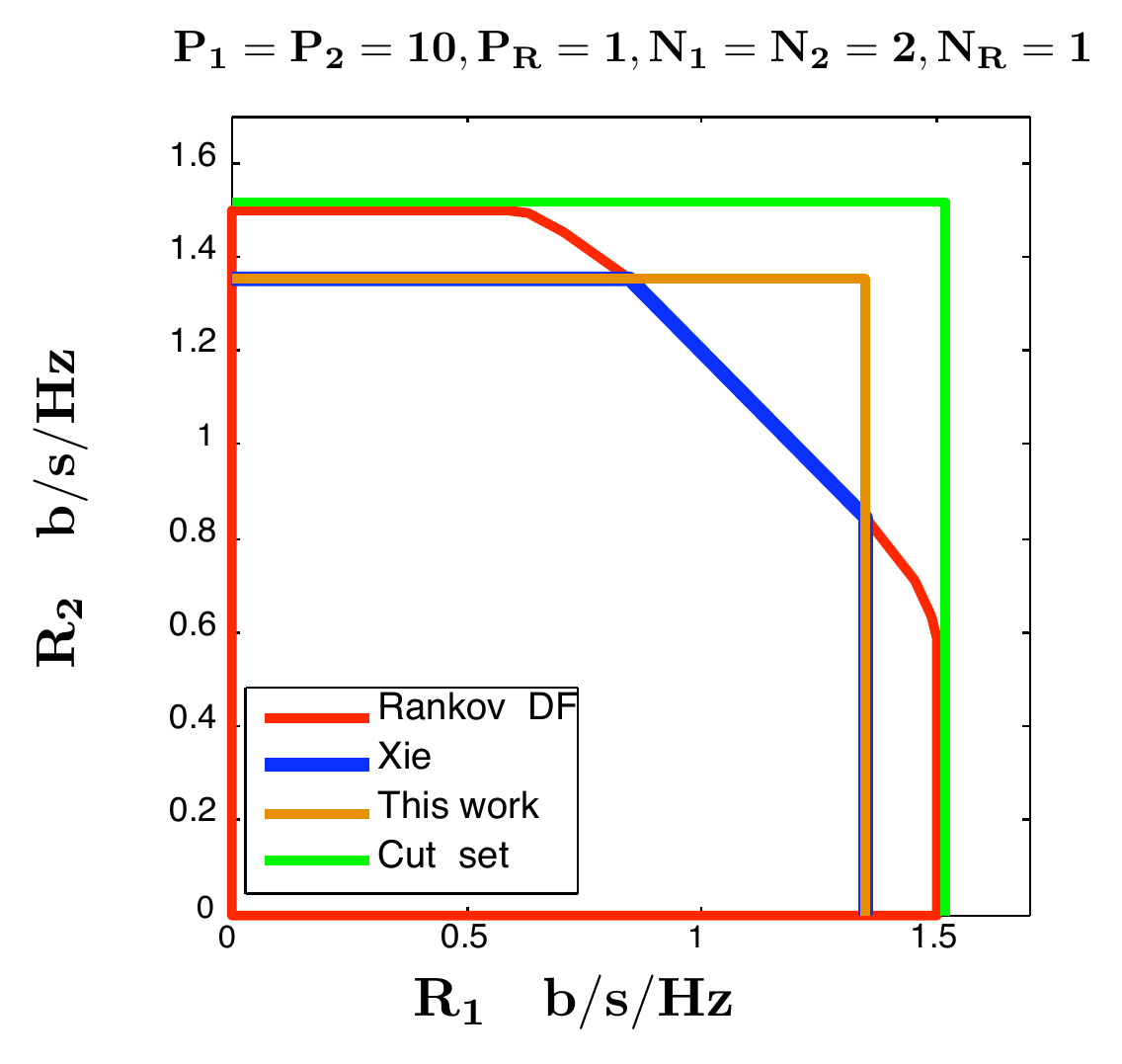}}
\subfigure{\includegraphics[width=2in]{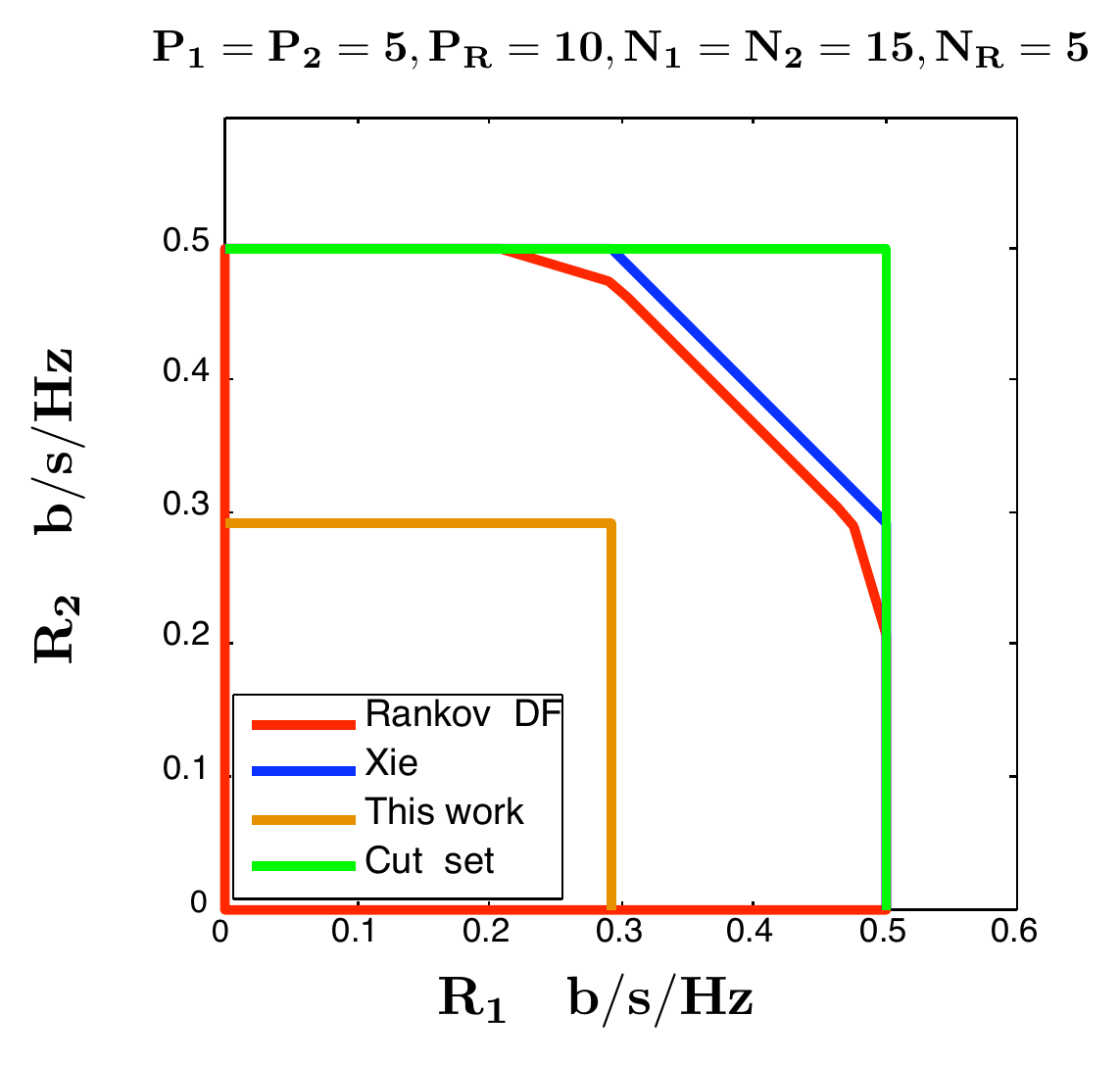}}
\caption{Comparison of decode-and-forward achievable rate regions of various two-way relay channel rate regions.}
\label{fig:numerical}
\end{figure*}

\subsection{The multiple-access relay channel}

We now consider a second example  of a relay network with two messages and cooperative relay links: the multiple-access relay channel (MARC).  The MARC was proposed and studied in \cite{Kramer:Gastpar:Gupta, sankar2007offset, woldegebreal2007multiple}, and describes a multi-user communication scenario in which two users transmit different messages to the same destination with the help of a relay. As in the TWRC, the MARC can be seen as another example of an extension of the three-node relay channel. The channel model is described by
\begin{align*}
{\bf Y_R} &= {\bf X_1} + {\bf X_2} + {\bf Z_R},  \;\;\;\; { \bf Z_R} \sim {\cal N}({\bf 0},N_R{\bf I}) \\
{\bf Y_D} &= {\bf X_1} + {\bf X_2} + {\bf X_D} + {\bf Z_D} ,  \;\;\;\; { \bf Z_D} \sim {\cal N}({\bf 0},N_D{\bf I}).
\end{align*}
where ${\bf X_1}$, ${\bf X_2}$ and ${\bf X_R}$ have power constraints $P_1$, $P_2$ and $P_R$.

An $(2^{nR_1}, 2^{nR_2}, n)$ code for the multiple access relay channel consists of the two sets of messages $w_i$, $i=1,2$ uniformly distributed over  ${\cal M}_i : = \{1,2,\cdots 2^{nR_i}\}$,  and two encoding functions $X_i^n: {\cal M}_i \rightarrow {\mathbb R}^n$ (shortened to ${\bf X_i}$) satisfying the power constraints $P_i$, a set of relay functions $\{f_j\}_{j=1}^n$ such that the relay channel input at time $j$ is a function of the previously received relay channel outputs from channel uses $1$ to $j-1$, $X_{R,j} = f_j(Y_{R,1}, \cdots Y_{R,j-1})$, and one decoding functions $g: {\cal Y}^n \rightarrow {\cal M}_1 \times {\cal M}_2$ which yields the message estimates $(\hat{w}_1, \hat{w}_2): = g(Y^n)$. 
We define the average probability of error of the code to be $P_{n,e} : = \frac{1}{2^{n(R_1+R_2)}} \sum_{w_1\in {\cal M}_1, w_2\in {\cal M}_2} \Pr\{(\hat{w_1}, \hat{w_2}) \neq (w_1,w_2)|(w_1,w_2) \mbox{ sent}\}$. The rate pair $(R_1,R_2)$ is then said to be achievable by the multiple access relay channel if, for any $\epsilon>0$ and for sufficiently large $n$, there exists an $(2^{nR_1},2^{nR_2},n)$ code such that $P_{n,e} < \epsilon$. The capacity region $C$ of the multiple access relay channel is the supremum of the set of achievable rate pairs. 

We derive a new achievable rate region whose achievability scheme combines the previously derived lattice DF scheme,  and the linearity of lattice codes using lattice list decoding. In particular, we demonstrate how we may decode the sum of two lattice codewords at the relay rather than decoding the individual messages, eliminating the sum-rate constraint seen in i.i.d. random coding schemes. 
The relay then forwards a re-encoded version of this which may be combined with lattice list decoding at the destination to obtain a new rate region.

\begin{theorem}
{\it Lattices in the AWGN multiple access relay channel.} For any $\alpha\in [0,1]$, the following rates are achievable for the AWGN multiple access relay channel:

\begin{align*}
R_1 < \alpha \min \left( \left[\frac{1}{2} \log \left( \frac{P_1}{P_1 + P_2} + \frac{P_1}{N_R}\right)\right]^+, \frac{1}{2} \log \left(1 + \frac{P_1}{P_2 +P_R +N_D}\right) \right) \\
+ \, (1-\alpha) \min \left( \left[\frac{1}{2} \log \left( \frac{P_1}{P_1 + P_2} + \frac{P_1}{N_R}\right)\right]^+, \frac{1}{2} \log \left(1 + \frac{P_1+P_R}{N_D}\right)\right), \\
R_2 < (1-\alpha) \min \left(\left[\frac{1}{2} \log \left( \frac{P_2}{P_1 + P_2} + \frac{P_2}{N_R}\right)\right]^+, \frac{1}{2} \log \left(1 + \frac{P_2}{P_1 +P_R +N_D}\right) \right) \\
+ \alpha \min \left(\left[\frac{1}{2} \log \left( \frac{P_2}{P_1 + P_2} + \frac{P_2}{N_R}\right)\right]^+, \frac{1}{2} \log \left(1 + \frac{P_2+P_R}{N_D}\right)\right).
\end{align*}


\label{thm:MARC}
\end{theorem}

\begin{proof}

\begin{figure*}
\centering
\includegraphics[width=14cm]{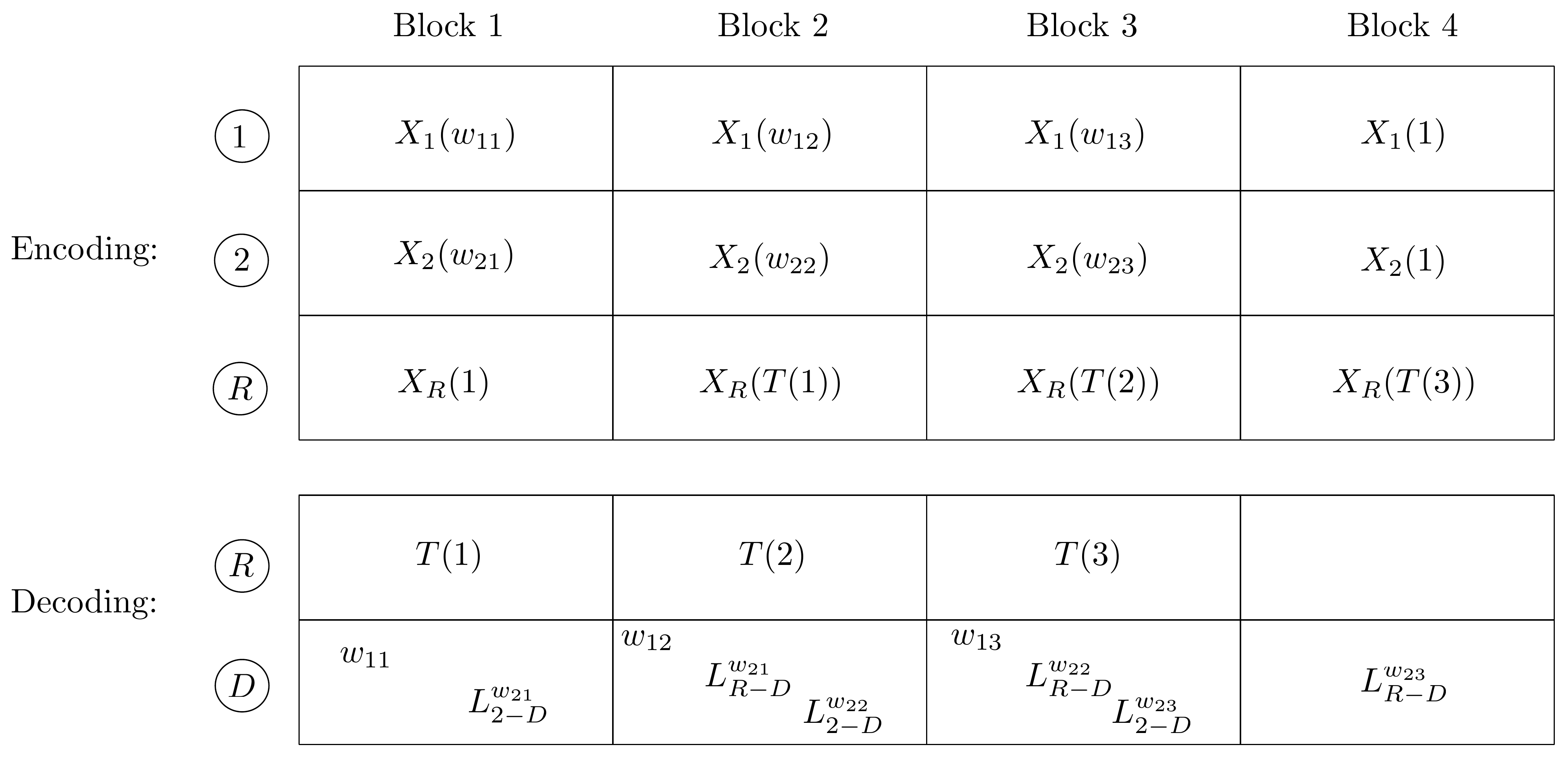}
\caption{Lattice Decode-and-Forward scheme for the AWGN multiple access relay channel. }
\label{DF for multipleaccess}
\end{figure*}

{\bf Codebook construction:} 
We construct two nested lattice chains according to Theorem \ref{thm:nam},  $\Lambda_1, \Lambda_2, \Lambda_{s1}, \Lambda_{s2}, \Lambda_{c1}, \Lambda_{c2}$ and $\Lambda_R, \Lambda_{sR1}, \Lambda_{sR2}, \Lambda_{cR}$,  nested in the exact same way as in the codebook construction of Theorem \ref{thm:two-way}. 

{\bf Encoding:} We again use block Markov encoding. At the $b$-th block,  terminal 1 and 2  send ${\bf X_1}(w_{1b}) $ and ${\bf X_2}(w_{2b})$, where
\begin{align*}
 {\bf X_1}(w_{1b}) &= ({\bf t_1}(w_{1b}) - {\bf U_1}(b)) \mod \Lambda_1\\
{\bf X_2}(w_{2b}) &= ({\bf t_2}(w_{2b}) - {\bf U_2}(b)) \mod \Lambda_2. 
\end{align*}
 At the relay,  we assume that it has decoded 
 \[ {\bf T}(b-1) = ({\bf t_1}(w_{1(b-1)}) + {\bf t_2}(w_{2(b-1)}) - Q_2({\bf t_2}(w_{2(b-1)}) + {\bf U_2}(b-1)) \mod \Lambda_1\]
  in block $b-1$. 
  Following the exact same steps as in between \eqref{eq:T1} and \eqref{eq:T2}, 
 the relay sends
\[ {\bf X_R}({\bf T}(b-1)) = ({\bf t_R}( {\bf T} (b-1)) - {\bf U_R}(b-1)) \mod \Lambda_R.\]
The dithers ${\bf U_1}(b), {\bf U_2}(b)$, and ${\bf U_R}(b)$ are known to all nodes and are i.i.d. and uniformly distributed over ${\cal V}_1$, ${\cal V}_2$, and ${\cal V}_R$ and vary from block to block.
In the first block 1, terminal 1 and terminal 2 send ${\bf X_1}(w_{11})$ and ${\bf X_2}(w_{21})$ respectively, while the relay sends a known ${\bf X_R}(1)$.

{\bf Decoding:} At the end of each block $b$, the relay terminal receives ${\bf Y_R}(b) = {\bf X_1}(w_{1b}) + {\bf X_2}(w_{2b}) + {\bf Z_R}(b)$ and decodes ${\bf T}(b) = ({\bf t_1}(w_{1b}) + {\bf t_2}(w_{2b}) - Q_2({\bf t_2}(w_{2b}) + {\bf U_2}(b)) \mod \Lambda_1 $ as long as
{ \begin{align*}
R_1 &\leq  \left [ \frac{1}{2} \log \left( \frac{P_1}{P_1 + P_2} + \frac{P_1}{N_R}\right) \right ]^+,  \;\; R_2 \leq \left [ \frac{1}{2} \log \left( \frac{P_2}{P_1 + P_2} + \frac{P_2}{N_R}\right) \right ]^+
\end{align*}}
following arguments similar to those in \cite{Nam:IEEE}.

The destination receives ${\bf Y_D}(b) = {\bf X_1}(w_{1b})  + {\bf X_2}(w_{2b}) +  {\bf X_R}( {\bf T}(b-1) ) + {\bf Z_D}(b)$ and either decodes the messages in the order $w_{1b}$ and then $w_{2(b-1)}$ or the reverse $w_{2b}$ and then $w_{1(b-1)}$. We describe the former; the latter follows analogously and we time-share between the two decoding orders. 
The destination first decodes $w_{1b}$ from ${\bf Y_D}(b)$, treating ${\bf X_2}(w_{2b}) +  {\bf X_R}({\bf T}(b-1)) + {\bf Z_D}(b)$ as noise. This equivalent noise is the sum of signals uniformly distributed over fundamental Voronoi  regions of Rogers good lattices and Gaussian noise. Hence, 
 according to Lemma \ref{lem:unique}, the probability of error in decoding the correct (unique) $w_{1b}$ will decay exponentially as long as
\[ R_1 < C \left( \frac{P_1}{P_2 + P_R + N_D}  \right).\]

It then subtracts ${\bf X_1}(w_{1b})$ from the signal ${\bf Y_D}(b)$ to obtain 
${\bf Y_D^*}(b) =  {\bf X_2}(w_{2b}) +  {\bf X_R}( {\bf T}(b-1) ) + {\bf Z_D}(b)$
 and decodes a list of $w_{2(b-1)}$ denoted by $L_{R-D}^{w_{2(b-1)}}({\bf Y_D^*}(b))$ of size $2^{n\left (R_2 - C\left( \frac{P_R}{P_2+ N_D} \right)\right)}$ 
 assuming side information $w_{1(b-1)}$, 
and treating ${\bf X_2}(w_{2b}) + {\bf Z_D}(b)$ as noise. This list of $w_{2(b-1)}$ is obtained from a lattice list decoder based on ${\bf t_R}({\bf T}(b-1))$ and noting the one-to-one correspondence between ${\bf t_R}({\bf T}(b-1))$ and ${\bf t_2}(w_{2(b-1)})$ and hence $w_{2(b-1)}$ given ${\bf t_1}(w_{1(b-1)}$, using the arguments of \eqref{eq:equiv1} and \eqref{eq:equiv2}. 

The destination then intersects the list $L_{R-D}^{w_{2(b-1)}}({\bf Y_D^*}(b))$ with another list $L_{2-D}^{w_{2(b-1)}}({\bf Y_D^*}(b-1))$ of size $2^{n\left(R_2 - C \left(\frac{P_2}{N_D} \right) \right) }$ obtained in the block $b-1$ (described next for block $b$) to determine the unique $w_{2(b-1)}$.  Once the destination has decoded $w_{1b}$, $w_{2(b-1)}$ and $w_{1(b-1)}$, it is also able to reconstruct ${\bf X_R}({\bf T}(b-1))$.

At last, the destination decodes a list $L_{2-D}^{w_{2b}}({\bf Y_D^*}(b))$ of possible $w_{2b}$ of size $2^{n\left(R_2 - C \left(\frac{P_2}{N_D} \right) \right) }$ 
from the signal ${\bf Y_D^*}(b)  = {\bf Y_D}(b)- {\bf X_1} (w_{1b}) - {\bf X_R} ({\bf T}(b-1)) = {\bf X_2} (w_{2b}) + {\bf Z_D} (b) $ which is used to determine $w_{2b}$ in 
the next block $b+1$. To ensure that there is an unique codeword $w_{2(b-1)}$ in the intersection
 of the two lists $L_{R-D}^{w_{2(b-1)}}({\bf Y_D^*}(b))$ and $L_{2-D}^{w_{2(b-1)}}({\bf Y_D^*}(b-1))$, we need
\begin{align*}
R_2 &<  C\left( \frac{P_R}{P_2+ N_D} \right) + C\left(\frac{P_2}{N_D}\right) = \frac{1}{2} \log  \left(1 +  \frac{P_2 + P_R}{N_D} \right).
\end{align*}
We presented the decoding order $w_{1b}, w_{2(b-1)}$. Alternatively, one may decode in the order $w_{2b}$ and $w_{1(b-1)}$ at the analogous rates. 
Time sharing with parameter $0\leq \alpha\leq 1$ between the orders yields the theorem.

\end{proof}
\begin{remark}
Note that the above region is derived using time-sharing between two decoding orders at the destination. This results as we employ successive decoding at the destination in order to allow for the use of lower complexity Euclidean lattice decoding, rather than a more complex form of ``joint'' decoding for lattices proposed for example in \cite{Ozgur:2010:lattice, ordentlich2011interference}. Further note that this region does not always outperform or even attain the same rates as random coding based schemes -- in fact, as in the two-way relay channel, there is a  trade off between rate gains from decoding the sum at the relay node, and coherent gains and joint decoding at the destination.
\end{remark}

\section{Single source Compress and Forward}
\label{sec:CF1}



We have shown several lattice based Decode-and-Forward schemes for relay networks. 
Forcing the relay(s) to decode the message(s)  they do not  need imposes a rate constraint; Compress-and-Forward (CF) is an alternative type of forwarding which alleviates this constraint.  Cover and El Gamal first proposed a CF scheme for the relay channel in \cite{Cover:1979} in which the relay does not decode the message but instead compresses its received signal and forwards the compression index. The destination first recovers the compressed signal, using its direct-link  side-information (the Wyner-Ziv problem  of lossy source coding with correlated side-information at the receiver), and then proceeds to decode the message from the recovered  compressed signal and the received signal. 

It is natural to wonder whether lattice codes may be used in the original Cover and El Gamal CF scheme for the relay channel.  We answer this in the positive.
We note that lattices have recently been shown to achieve the Quantize-Map-and-Forward rates for general relay channels using Quantize-and-Map scheme (similar to the CF scheme) which quantizes the received signal at the relay and re-encodes it without any form of binning / hashing in \cite{Ozgur:2010:lattice}. The contribution in this section 
is to show an alternative achievability scheme which achieves the same rate in the three node relay channel, demonstrating that lattices may be used to achieve CF-based rates in a number of fashions. We note that our decoder employs a lattice decoder rather than the more complex joint typicality, or ``consistency check'' decoding of \cite{Ozgur:2010:lattice}.

In the CF scheme of \cite{Cover:1979}, Wyner-Ziv coding -- which exploits binning --  is used at the relay to exploit receiver side-information obtained from the direct link between the source and destination. The usage of lattices and structured codes for binning (as opposed to their random binning counterparts) was considered in a comprehensive fashion  in \cite{Zamir:2002:binning}. Of particular interest to the problem considered here is the nested lattice-coding approach of \cite{Zamir:2002:binning} to the Gaussian Wyner-Ziv coding problem. 

\subsection{Lattice codes for the Wyner-Ziv model in Compress-and-Forward}
\label{sec:WZ}
We consider the lossy compression of the Gaussian source ${\bf Y} = {\bf X} +{\bf Z_1}$ , with Gaussian side-information ${\bf X} + {\bf Z_2}$ available at the reconstruction node, where ${\bf X} , {\bf Z_1}$ and ${\bf Z_2}$ are independent vectors of length $n$ which are independent and each generated in an i.i.d. fashion according to a Gaussian of zero mean and variance $P, N_1$, and $N_2$, respectively. We use the same definitions for the channel model and for achievability as in Section \ref{subsec:DF}. 
The rate-distortion function for the source ${\bf X} + {\bf Z_1}$ taking on values in ${\cal X}_1^n = \mathbb{R}^n$ with side-information ${\bf X} + {\bf Z_2}$ taking on values in ${\cal X}_2^n = \mathbb{R}^n$ 
is the infimum of rates $R$ such that there exist maps $i_n: {\cal X}_1^n  \rightarrow \{1,2,\cdots, 2^{nR}\}$ and $g_n: {\cal X}_2^n \times \{1,2, \cdots, 2^{nR}\}\rightarrow {\cal X}_1^n$ such that $\lim \sup_{n\rightarrow \infty} E[d({\bf X} + {\bf Z_1} , g_n({\bf X} + {\bf Z_2}, i_n({\bf X} +{\bf Z_1}))]\leq D$ for some distortion measure $d(\cdot, \cdot)$.
 If the distortion measure $d(\cdot, \cdot)$ is the squared error distortion, $d({\bf X},\widehat{{\bf X}}) = \frac{1}{n}E[|| {\bf X}-\widehat{{\bf X}}||^2]$, then, by  \cite{wyner1978rate},
  the rate distortion function $R(D)$ for the source $ X +  Z_1$ given the side-information $  X + Z_2$ is given by
\begin{align*}
R(D) &= \frac{1}{2} \log \left(\frac{\sigma^2_{ X +  Z_1 |  X +  Z_2}}{D} \right),  \qquad 0 \leq D \leq \sigma^2_{ X  + Z_1| X + Z_2}\\
&= \frac{1}{2} \log \left( \frac{N_1 + \frac{PN_2}{P+N_2}}{D} \right),  \qquad 0 \leq D \leq N_1 + \frac{PN_2}{P+N_2},
\end{align*}
and $0$ otherwise,  where $\sigma^2_{ X +   Z_1 | X +  Z_2}$ is the conditional variance of $X+  Z_1$ given $X + Z_2$.

A general lattice code implementation of the Wyner-Ziv scheme is considered in \cite{Zamir:2002:binning}. 
In order to mimic the CF scheme achieved by Gaussian random codes of \cite{Cover:1979}, we need a slightly sub-optimal version  of the 
optimal scheme described in \cite{Zamir:2002:binning}. That is, 
in the context of CF, and to mimic the rate achieved by independent Gaussian random codes used for compression in the CF rate of \cite{Cover:1979},
the quantization noise after compression should be independent of the signal to be compressed to allow for two independent views of the source, i.e. to express the compressed signal as ${\bf \hat{Y}} = {\bf Y} - {\bf E_q} = {\bf X} + {\bf Z_1} -  {\bf E_q}$ where ${\bf E_q}$ is independent of ${\bf X} + {\bf Z_1}$. 
This may be achieved by selecting $\alpha_1=1$ in a modified version of the lattice-coding Wyner-Ziv scheme of \cite{Zamir:2002:binning} rather than  the optimal MMSE scaling coefficient  $\alpha_1  = \sqrt{1 - \frac{D}{N_1 + \frac{P N_2}{P+N_2}}}$. 
This roughly allows one to view ${\bf \widehat{Y}} = {\bf X} + {\bf N_1} -{\bf  E_q}$  as an equivalent AWGN channel, and is the form generally used in Gaussian CF as in \cite{Gammal:LN}. Whether this is optimal is unknown.
The second difference from direct application of \cite{Zamir:2002:binning} is that, in our lattice CF scheme, the signal ${\bf X}$ is no longer Gaussian but uniformly distributed over the fundamental Voronoi  region of a Rogers good lattice. We modify the scheme of  \cite{Zamir:2002:binning} to incorporate these two changes next.


\begin{figure*}[ht]
\centering
\includegraphics[width=16cm]{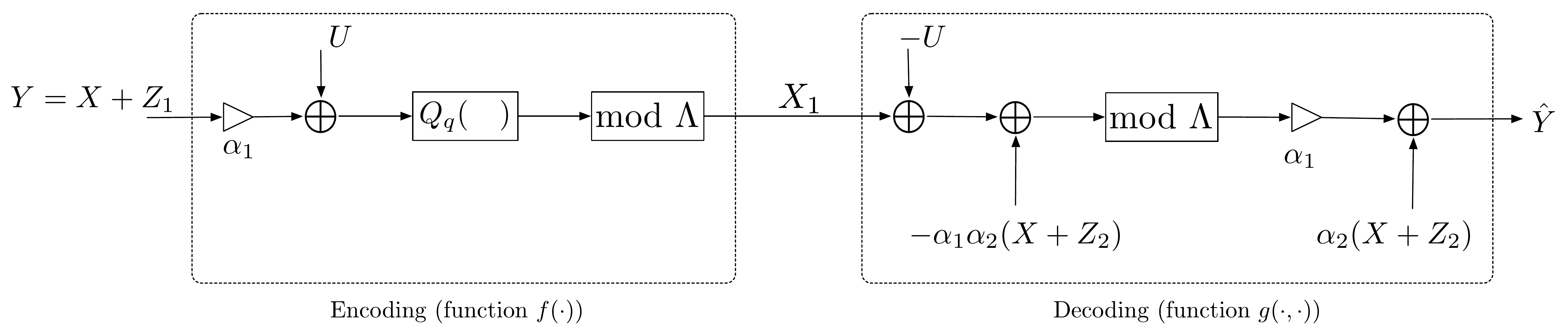}
\caption{Lattice coding for the $( {\bf X} +{\bf Z_1}, {\bf X} +{\bf Z_2})$ Wyner-Ziv problem.}
\label{fig:WZ}
\end{figure*}

\corollary{{\it Lattices for the $( {\bf X} +{\bf Z_1}, {\bf X} + {\bf Z_2})$ Wyner-Ziv problem used in the lattice CF scheme based on \cite{Zamir:2002:binning}. }
Let ${\bf X}$ be uniformly distributed over the fundamental Voronoi  region of a Rogers good lattice with second moment $P$, while ${\bf Z_1}\sim {\cal N}({\bf 0}, N_1 {\bf I})$ and ${\bf Z_2}\sim {\cal N}({\bf 0}, N_2 {\bf I})$. The following rate-distortion function for the lossy compression of the source ${\bf X} +{\bf Z_1}$ to be reconstructed as ${\bf X} + {\bf Z_1} - {\bf E_q}$ (where ${\bf E_q}$ is independent of ${\bf X} + {\bf Z_1}$ and has variance $D$) may be achieved using lattice codes:
\[ R(D)
= \frac{1}{2} \log \left(1+ \frac{N_1 + \frac{PN_2}{P+N_2}}{D} \right),  \qquad 0 \leq D \leq \infty. \]
\label{thm:WZ}

}
\begin{proof}
Consider a pair of nested lattice codes $\Lambda \subseteq \Lambda_q$, where $\Lambda_q$ is Rogers-good with second moment $D$, and $\Lambda$ is Poltyrev-good with second moment $N_1 + \frac{PN_2}{P + N_2}+D$. The existence of such a nested lattice pair good for quantization is guaranteed as in \cite{Zamir:2002:binning}. We consider the encoding and decoding schemes of Figure \ref{fig:WZ}, similar to that of \cite{Zamir:2002:binning}. We let ${\bf U}$ be a quantization dither signal which is uniformly distributed over $\mathcal{V}_q$ and introduce  the following coefficients (choices justified later):
\begin{equation}
 \alpha_1  = 1,
 \;\;\;\; \alpha_2= \frac{P}{P+ N_2}.
 \label{eq:alpha}
 \end{equation}

{\bf Encoding:} The encoder quantizes the scaled and dithered signal $\alpha_1( {\bf X} +{\bf Z_1})+{\bf U}$ to the nearest fine lattice point, which is then modulo-ed back to the coarse lattice fundamental Voronoi  region as
\begin{align*}
{\bf I} &:= Q_q( \alpha_1 ({\bf X} +{\bf Z_1}) + {\bf U} ) \mod \Lambda \\
&= ({\bf X}+ {\bf Z_1}+ {\bf U} - {\bf E_q}) \mod \Lambda.
\end{align*}
where ${\bf E_q} := ( {\bf X} +{\bf Z_1} + {\bf U}) \mod \Lambda_q$ is independent of ${\bf X}+{\bf Z_1}$ and uniformly distributed over $\mathcal{V}_q$ according to the Crypto lemma \cite{Forney:ShannonWinener:2003}.  The encoder sends index $i$ of ${\bf I}$ at the source coding rate
\begin{align*}
R &= \frac{1}{n} \log \left( \frac{V(\Lambda)}{V(\Lambda_q)} \right)
 = \frac{1}{2} \log \left( 1+ \frac{N_1 + \frac{PN_2}{P+N_2}}{D} \right).
\end{align*}

{\bf Decoding:} The decoder receives the index $i$ of ${\bf I}$ and reconstructs ${\bf \widehat{Y}}$ as
\begin{align*}
{\bf \widehat{Y} }&=  \alpha_1( ( {\bf I} - {\bf U} - \alpha_1 \alpha_2 ({\bf X} + {\bf Z_2}) ) \mod \Lambda ) + \alpha_2( {\bf X} + {\bf Z_2}) \\
&= \alpha_1 ( ( \alpha_1( (1-\alpha_2) {\bf X} - \alpha_2 {\bf Z_2}  + {\bf Z_1}) - {\bf E_q} ) \mod \Lambda ) + \alpha_2( {\bf X} + {\bf Z_2})  \\
&\overset{(a)}{\equiv} \alpha_1 ( \alpha_1( (1 - \alpha_2) {\bf X} - \alpha_2 {\bf Z_2} + {\bf Z_1}) - {\bf E_q} )  + \alpha_2({\bf X} + {\bf Z_2}) \\
& = {\bf X} + {\bf Z_1} - {\bf E_q}
\end{align*}
where equivalence (a) denotes asymptotic equivalence (as $n\rightarrow \infty$), since, as in  \cite[Proof of (4.19)]{Zamir:2002:binning}
\begin{align}
\Pr \{ ( &\alpha_1( (1-\alpha_2) {\bf X} - \alpha_2 {\bf Z_2}  + {\bf Z_1}) - {\bf E_q} ) \mod \Lambda \neq  \alpha_1( (1-\alpha_2) {\bf X} - \alpha_2 {\bf Z_2} + {\bf Z_1}) - {\bf E_q}  \}  \rightarrow 0 \label{pe}
 \end{align}
for a sequence of a good nested lattice codes since
\begin{align}
\frac{1}{n} &E || \alpha_1( (1-\alpha_2) {\bf X} - \alpha_2 {\bf Z_2} + {\bf Z_1}) - {\bf E_q} ||^2  =   \frac{PN_2}{P + N_2} + N_1 + D =  \sigma^2(\Lambda).  \label{eq:imp}
\end{align}
Note that there is a slight difference from \cite[Proof of (4.19)]{Zamir:2002:binning} since ${\bf X}$ is uniformly distributed over the fundamental Voronoi  region of a Rogers good lattice rather than Gaussian distributed. However,
according to Lemma \ref{lem:gaussianbound}, $\alpha_1 ( (1-\alpha_2) {\bf X_1} - \alpha_2 {\bf Z_2} + {\bf Z_1} )- {\bf E_q}  = (1-\alpha_2) {\bf X_1} - \alpha_2 {\bf Z_2} + {\bf Z_1} - {\bf E_q}$ may be upper bounded by the pdf of an i.i.d. Gaussian random vector (times a constant) with variance approaching  \eqref{eq:imp}
since ${\bf X_1}$ is uniformly distributed over the Rogers good ${\cal V}_1$, ${\bf E}_q$ is uniformly distributed over the Rogers good  ${\cal V}_q$ of second moment $D$, and $- \alpha_2 {\bf Z_2} + {\bf Z_1}$ is Gaussian. 
Then because $\Lambda$ is Poltyrev good,  \eqref{pe} can be made arbitrary small as $n \rightarrow \infty$. This guarantees a distortion of $D$ as ${\cal V}_q$ is of second moment $D$. 

\end{proof}

\subsection{Lattice coding for Compress-and-Forward}
\label{sec:LCF}

Armed with a lattice Wyner-Ziv scheme, we mimic every step of the CF scheme for the AWGN relay channel of Figure \ref{fig:relay-channel} and Section \ref{subsec:DF}, described in \cite{Cover:1979} using lattice codes and will show that the same rate as that achieved using random Gaussian codebooks may be achieved in a structured manner.

\theorem{{\it Lattices achieve the CF rate for the relay channel.}
The following rate may be achieved for the AWGN relay channel using lattice codes in a lattice Compress-and-Forward fashion:
\begin{align*}
R & < \frac{1}{2} \log \left( 1 + \frac{P}{N_D} + \frac{PP_R}{PN_R + PN_D + P_RN_R + N_RN_D} \right).
\end{align*}
\label{thm:LCF}
}
\begin{proof}

\noindent {\bf Lattice codebook construction:}  We employ three nested lattice pairs of dimension $n$ satisfying:
\\ \noindent $\bullet$ Channel codebook for Node $S$: codewords ${\bf t_1} \in \mathcal{C}_1 = \{ \Lambda_{c1} \cap \mathcal{V}_1\}$ where $\Lambda_1 \subseteq \Lambda_{c1} $, and $\Lambda_1$ is both Rogers-good and Poltyrev-good and $\Lambda_{c1}$ is Poltyrev-good.  We set $\sigma^2(\Lambda_1) = P$ to satisfy the transmitter  power constraint. We associate each message $w\in \{1,2,\cdots 2^{nR}\}$ with a codeword ${\bf t_1}(w)$ in one-to-one fashion and send a dithered version of ${\bf t_1}(w)$. Note that $\Lambda_{c1}$ is chosen such that $|\mathcal{C}_1| = 2^{nR}$.
\\ \noindent $\bullet$  Channel codebook for the relay: codewords ${\bf t_R}\in \mathcal{C}_R = \{ \Lambda_{cR} \cap \mathcal{V}_R\}$ where $\Lambda_R \subseteq \Lambda_{cR} $, and  $\Lambda_R$ is both Rogers-good and Poltyrev-good and $\Lambda_{cR}$ is Poltyrev-good. We set $\sigma^2(\Lambda_R) = P_R$ to satisfy the relay power constraint. We associate each compression index $i\in \{1,2,\cdots, 2^{nR'}\}$ with the codeword ${\bf t_R}(i)$ in a one-to-one fashion and send a dithered version of ${\bf t_R}(i)$. Note that $\Lambda_{cR}$ is chosen such that $| \mathcal{C}_R| = 2^{nR'} $.
\\ \noindent $\bullet$  Quantization/Compression codebook: ${\bf t_q} \in \mathcal{C}_q = \{ \Lambda_{q} \cap \mathcal{V}\}$ and $\Lambda \subseteq \Lambda_{q} $, where  $\Lambda$ is Poltyrev-good and $\Lambda_q$ is Rogers-good.  We set $\sigma^2(\Lambda_q) = D$, $\sigma^2(\Lambda) = N_R + \frac{P_1N_2}{P_1 + N_2} + D$, such that the source coding rate is $\widehat{R} = \frac{1}{n} \log \left( \frac{V(\Lambda)}{V(\Lambda_q)} \right) = \frac{1}{2} \log \left( 1 + \frac{N_R + \frac{PN_D}{P + N_D}}{D} \right) $. 

\noindent
{\bf Encoding:} We use block Markov encoding as \cite{Cover:1979}. In block $b$,  Node 1 transmits 
\[ {\bf X_S}(w_{b}) = ( {\bf t_1}(w_b) + {\bf U_1}(b) ) \mod \Lambda_1, \]
where ${\bf U_1}(b)$ is the dither  uniformly distributed over $\mathcal{V}_1$.
The relay quantizes the received signal in the previous block $b-1, \; {\bf Y_R}(b-1) = {\bf X_S}(w_{b-1}) + {\bf Z_R}(b-1)$
 to ${\bf I}(w_{b-1}) = Q_q\left( {\bf X_S}(w_{b-1})+{\bf Z_R}(b-1) + {\bf U}_q - {\bf E}_q \right) \mod \Lambda$ (with index $i(w_{b-1})$) by using the quantization lattice code pair $(\Lambda_q, \Lambda)$ as described in the encoding part of Theorem \ref{thm:WZ}, for ${\bf U}_q$ a quantization dither uniformly distributed over ${\cal V}_q$ and ${\bf E}_q : =\left( {\bf X_S}(w_{b-1}) + {\bf Z_R}(b-1) + {\bf U}_q \right) \mod \Lambda_q$. 
Node 2 chooses the codeword ${\bf t_R}(i(w_{b-1}))$ associated with the index $i(w_{b-1})$ of $ {\bf I}(w_{b-1})$ and sends
 \[{\bf X_R}(w_{b-1}) = ({\bf  t_R}(i(w_{b-1})) + {\bf U_R}(b-1) ) \mod \Lambda \]
with ${\bf U_R}(b-1)$ the dither signal uniformly distributed over $\mathcal{V}_R$ and independent across blocks.

\noindent
{\bf Decoding:} In block $b$, Node $D$ receives
\[ {\bf Y_D}(b) = {\bf X_S}(w_{b}) + {{\bf X_R}(w_{b-1})} + {\bf Z_D}(b). \]
It first decodes  $w_{b-1}$
using lattice decoding as in \cite{Erez:2004} or Lemma  \ref{lem:unique} as long as 
\[ R' < \frac{1}{2} \log \left( 1 + \frac{P_R}{P + N_D} \right).\]
We note that the source coding rate of ${\bf I}$, $\widehat{R}$
must be less than the channel coding rate $R'$, i.e.
\begin{align}
 \frac{1}{2} \log \left( 1 + \frac{N_R + \frac{PN_D}{P + N_D}}{D} \right) <   \frac{1}{2} \log \left( 1 + \frac{P_R}{P + N_D} \right), \label{D}
 \end{align}
 which sets a lower bound on the achievable distortion $D$.  
Node $D$ then may obtain 
\begin{align*}
{\bf Y_D'}(b) &= {\bf Y_D}(b) - {{\bf X_R}(w_{b-1})}= {\bf X_S}(w_{b}) + {\bf Z_D}(b)
\end{align*}
which is used as direct-link side-information in the next block $b+1$.  In the previous block, Node $D$ had also obtained ${\bf Y_D'}(b-1) = {\bf X_S}(w_{b-1}) + {\bf Z_D}(b-1)$. Combining this with ${\bf I}(w_{b-1})$, Node $D$ uses ${\bf Y_D'}(b-1)$ as side-information to reconstruct ${\bf \widehat{Y}_D}(b-1)$ as in the decoder of Theorem \ref{thm:WZ}. 


Thus, we see that the CF scheme employs the $( {\bf X} + {\bf Z_1}, {\bf X} + {\bf Z_2})$ Wyner-Ziv coding scheme of Section \ref{sec:WZ} 
where the source to be compressed at the relay is ${\bf X_S} + {\bf Z_R}$ and the side-information at the receiver (from the previous block) is ${\bf X_S}+ {\bf Z_D}$.

The compressed ${\bf Y_R}(b-1)$ may now be expressed as
\begin{align*}
{\bf \widehat{Y}_R}(b-1)& = ( \alpha_1^2 - \alpha_1^2\alpha_2 + \alpha_2) {\bf X_S}(w_{b-1})  + \alpha_2 ( 1 - \alpha_1^2) {\bf Z_D}+ \alpha_1^2 {\bf Z_R}- \alpha_1 {\bf E_q}(b-1) \\
&= {\bf X_S}(w_{b-1}) + {\bf Z_R}(b-1) - {\bf E_q}(b-1)
\end{align*}
where ${\bf E_q}(b-1): = ({\bf Y_D}(b-1)+{\bf U_q}(b-1)) \mod \Lambda_q$  (with ${\bf U_q}(b-1)$ the quantization dither which is uniformly distributed over ${\cal V}_q$) is independent and uniformly distributed over $\mathcal{V}_q$ with second moment $D$.
The destination may decode ${\bf t_1}(w_{b-1})$ from ${\bf Y_D'}(b-1)$ and ${\bf \widehat{Y}_R}(b-1)$ by coherently combining them as
\begin{align}
&\frac{\sqrt{P}}{N_D}  {\bf Y_D'}(b-1)  +\frac{\sqrt{P}}{N_R + D} {\bf \widehat{Y}_R}(b-1) \nonumber \\
=& \left( \frac{\sqrt{P}}{N_D} + \frac{\sqrt{P}}{N_R + D} \right){\bf  X_S}(w_{b-1}) + \frac{\sqrt{P}}{N_D}{\bf  Z_D}(b-1) + \frac{\sqrt{P}}{N_R + D} \left({\bf Z_R}(b-1)- {\bf E_q}(b-1) \right). \label{CFY}
\end{align}
Now we wish to decode $w_{b-1}$ from \eqref{CFY} which is the sum of the desired codeword which is uniformly distributed over a Rogers good lattice, and noise composed of Gaussian noise and ${\bf E_q}$ uniformly distributed over a fundamental Voronoi  region of a Rogers good lattice. This 
scenario may be handled by Lemma \ref{lem:unique}, and we may thus uniquely decode $w_{b-1}$ as long as 
\[ R < \frac{1}{2} \log \left( 1 + \frac{P}{N_D} + \frac{P}{N_R + D} \right).\]
Combining this with the constraint \eqref{D}, we obtain
\begin{align*}
R &
< \frac{1}{2} \log \left( 1 + \frac{P}{N_D} + \frac{PP_R}{PN_R + PN_D + P_RN_R + N_RN_D} \right),
\end{align*}
which is the CF rate achieved by the usual choice of Gaussian random codes (in which the relay quantizes the received signal $\bf{Y_R}$ as $\bf{\hat{Y}_R} = \bf{Y_R} + \bf{E_q}$ in which $\bf{E_q}$ is independent of $\bf{Y_R}$) \cite[pg. 17--48]{Gammal:LN}.
\end{proof}

\medskip


\section{Conclusion}
\label{sec:conclusion}
We have demonstrated that lattice codes may mimic random Gaussian codes in the context of the Gaussian relay channel, achieving the  same Decode-and-Forward and Compress-and-Forward rates as those using random Gaussian codes. One of the central technical tools needed was a new lattice list decoder, which proved useful in networks with cooperation where various links to a destination carry different encodings of a given message. 
We have further demonstrated a technique for combining the linearity of lattice codes with classical Block Markov cooperation techniques in a DF fashion in two multi-source networks. Such achievability schemes outperform known i.i.d. random coding for certain channel conditions. 
   The question of whether lattice codes can replace random codes in all Gaussian relays networks and thereby achieve the same rates as the random coding counterparts remains open. 
%
Another remaining open question is whether the DF and CF schemes may be unified into a single scheme -- from the lattice DF and CF schemes presented here we notice that the relay performs a form of lattice quantization in both scenarios. 
Finally, the extension of these results  -- which roughly imply that structured codes may be used to replace random Gaussian codes in Gaussian networks  -- to discrete memoryless channels is of interest. In particular, structured codes such as ``abelian group codes'' \cite{abelian:2011} may prove useful in this direction. 

\begin{figure*}
\centering
\includegraphics[width=17cm]{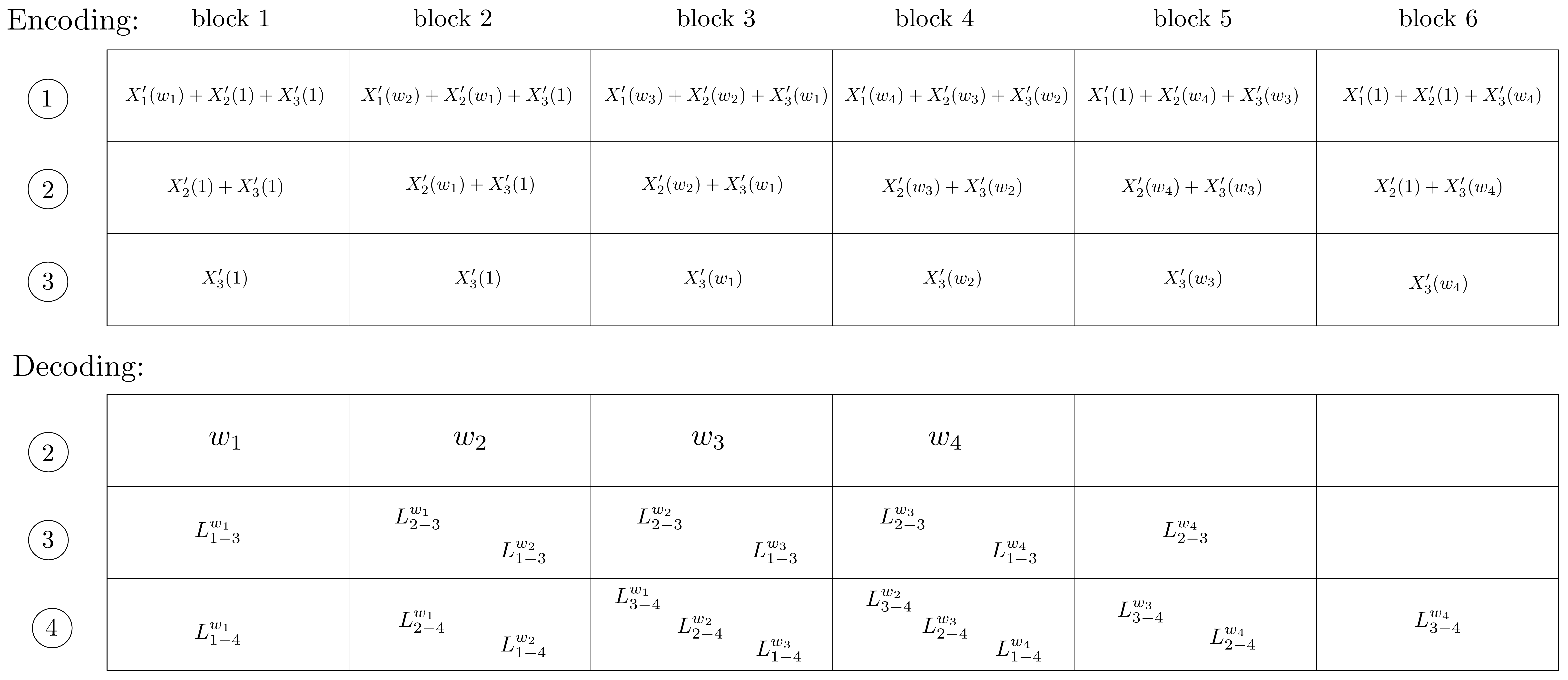}
\caption{Lattice Decode-and-Forward scheme for the AWGN multi-relay channel.}
\label{DF for multirelay}
\end{figure*}

\begin{appendix}

\subsection{Details in Decoding step 2. of Theorem \ref{thm:DF}.}
\label{app:list}

In applying the Lattice List Decoder of Theorem \ref{thm:list} to the steps between \eqref{eq:LRD} -- \eqref{eq:RR}, we form the list
\[ L_{R-D}^{w_{b-1}}({\bf Y_D}(b)) = \{w_{b-1} | \; {\bf {t_2}}(w_{b-1}) \in  S_{ \kappa \mathcal{V}_{s2},\kappa \Lambda_{c2}} ({\bf Y_D^\prime}(b)) \mod \kappa \Lambda_2\}, \]
where
\begin{align*}
{\bf Y^\prime_D}(b) &= ( \beta {\bf Y_D}(b) + \kappa{\bf U_2}({b-1})  ) \mod \kappa \Lambda_2 \\
&= (\kappa {\bf t_2}(w_{b-1}) - (1 - \beta) \kappa {\bf X_2'}(w_{b-1}) + \beta({\bf X'_1}(w_b) + {\bf Z_D}(b)) ) \mod \kappa \Lambda_2.
\end{align*}
 As in Section \ref{subsec:nested}, choose $\beta$ to be the MMSE coefficient $\beta_{MMSE}= \frac{\kappa^2 \bar{\alpha} P}{\kappa^2 \bar{\alpha} P+ \alpha P + N_D}$, resulting in self-noise ${\bf Z}_{eq}: = ((1 - \beta) \kappa {\bf X_2'}(w_{b-1}) + \beta({\bf X'_1}(w_b) + {\bf Z_D}(b)) ) \mod \kappa\Lambda_2$ of variance
 \[N_{eq} = \frac{\kappa^2 \bar{\alpha}P (\alpha P + N_D ) } { \kappa^2 \bar{\alpha}P+ \alpha P + N_D }.\]
Select $\Lambda_{s2}$ in the lattice chain  $\Lambda_2 \subseteq \Lambda_{s2}\subseteq \Lambda_{c2}$ to have a fundamental Voronoi region of volume $V_{s2} = \left( \frac{\alpha P + N_D}{\alpha P + N_D + (\sqrt{\bar{\alpha}P} + \sqrt{P_R})^2} \right) ^{n/2} V_2 $ {asymptotically} (notice $V_{s2}<V_2$ as needed). This will ensure a list of the desired size $2^{n(R-R_R)}$ as long as $R_R <   C((\sqrt{\bar{\alpha P}} + \sqrt{P_R})^2 / (\alpha P+N_D))$. 
For rates $R$ approaching $\frac{1}{2}\log\left(1+\frac{P+P_R+2\sqrt{\bar{\alpha}P P_R}}{N_D} \right) $ (where list decoding is needed / relevant),  $V_{c2} = \left( \frac{N_D}{ P + P_R + 2\sqrt{\bar{\alpha}PP_R} + N_D} \right)^{n/2} V_2$ asymptotically. Thus $V_{c2} < V_{s2} < V_2$ as needed.

\subsection{Proof of Theorem \ref{thm:DFm}}
\label{app:DFm}

\begin{proof}
Here we demonstrate achievability for the permutation $\pi(2) = 2, \pi(3)=3$, and thus drop $\pi(\cdot)$ to simplify notation. The other permutation may be analogously achieved. 
Source Node 1 transmits a message to the destination Node 4 with the help of two relays: Node 2 and Node 3.  The achievability scheme follows a generalization of the lattice regular encoding/sliding window decoding DF scheme of Theorem \ref{thm:DF}. The only difference is the addition of one relay and thus one coding level.  

\noindent{\bf Codebook construction:} 
We construct three nested lattice chains according to Theorem \ref{thm:nam}:
\begin{itemize}
\item $\Lambda_1\subseteq \Lambda_{s(1-3)}  \subseteq \Lambda_{s(1-4)} \subseteq \Lambda_{c1}$, or   $\Lambda_1\subseteq \Lambda_{s(1-4)}  \subseteq \Lambda_{s(1-3)} \subseteq \Lambda_{c1}$ (relative nesting order depends on the system parameters and will be discussed in the following paragraph)
\item  $\Lambda_2\subseteq \Lambda_{s(2-3)}  \subseteq \Lambda_{s(2-4)}  \subseteq \Lambda_{c2}$,  or $\Lambda_2\subseteq \Lambda_{s(2-4)}  \subseteq \Lambda_{s(2-3)}  \subseteq \Lambda_{c2}$
\item $\Lambda_3 \subseteq \Lambda_{s(3-4)}  \subseteq \Lambda_{c3}$
\end{itemize}
How these are ordered depends on the relative values of the power split parameters $\alpha_1, \beta_1, \alpha_2 \in [0,1]$, the power constraints $P_1,P_2,P_3$ and the noise variances $N_2,N_3,N_4$. 
In particular, the second moments of coarse lattices are selected as: $\sigma^2(\Lambda_1) = \alpha_1 P_1$, $\sigma^2(\Lambda_2) = \beta_1 P_1$, and $\sigma^2(\Lambda_3) = (1-\alpha_1 - \beta_1 ) P_1$.
The message set $w \in \{1,2,\cdots 2^{nR}\}$ is mapped in a one-to-one fashion  to three codebooks ${\bf t_1}(w) \in \mathcal{C}_1 = \{ \Lambda_{c1} \cap \mathcal{V}_1\}$,  ${\bf t_2}(w) \in \mathcal{C}_2 = \{ \Lambda_{c2} \cap \mathcal{V}_2\}$, and ${\bf t_3}(w) \in \mathcal{C}_3 = \{ \Lambda_{c3} \cap \mathcal{V}_3\}$. These mappings are independent. The fine lattices $\Lambda_{c1}, \Lambda_{c2}, \Lambda_{c3}$ may be chosen to satisfy the needed rate constraint $R$ by proper selection of the corresponding $\gamma$ in Theorem \ref{thm:nam}. 
 The lattices $\Lambda_{s(1-3)}, \Lambda_{s(2-3)}$ will be used for lattice list decoding at relay 3, while $\Lambda_{s(1-4)}$,  $\Lambda_{s(2-4)}$, and  $ \Lambda_{s(3-4)}$ will be used for lattice  list decoding at the destination node 4. They will all be Rogers good, with fundamental Voronoi region volume specified by the desired lattice list decoding constraints; we are able to select this volume (or equivalently second moment) arbitrarily as long as they are smaller than their corresponding nested coarse lattices, by Theorem \ref{thm:nam}. 
 In which order they are nested will depend on the relative volumes, which in turn depends on the systems parameters $\alpha_1, \beta_1, \alpha_2 \in [0,1]$, the power constraints $P_1,P_2,P_3$ and the noise variances $N_2,N_3,N_4$. 

Define the following signals (which will be superposed as described in the Encoding): 
\begin{align*}
{\bf X_1'}(w_b) &= ( {\bf t_1} (w_b) + {\bf U_1}(b) ) \mod \Lambda_1\\
{\bf X_2'}(w_b) &= ( {\bf t_2} (w_b) + {\bf U_2}(b) ) \mod \Lambda_2\\
{\bf X_3'}(w_b) &= ( {\bf t_3} (w_b) + {\bf U_3}(b) ) \mod \Lambda_3,
\end{align*} 
where ${\bf U_1}$, ${\bf U_2}$ and ${\bf U_3}$ are the dithers which are uniformly distributed over $\mathcal{V}_1$, $\mathcal{V}_2$ and $\mathcal{V}_3$, respectively, independent from block to block, and independent of each other.   The encoding and decoding steps are outlined in Figure \ref{DF for multirelay}. We make a small remark on our notation: ${\bf X_i'}$ should not be thought of as the signal being transmitted by Node $i$ (which would be ${\bf X_i}$ but we do not use this, opting instead to write out the transmit signals in terms of ${\bf X_i'}$). Rather, Node $i$ will send a superposition of the signals ${\bf X_i'}, {\bf X_{i+1}'}, \cdots$. Thus, multiple nodes may transmit the same (scaled) codeword ${\bf X_i'}$ which will coherently combine. 

\smallskip
\noindent{\bf Encoding:} We again use block Markov encoding: the message is divided into B blocks of $nR$ bits each. In block $b$, suppose Node 2 knows $\{ w_1, \dots, w_{b-1} \}$ and Node 3 knows $\{ w_1, \dots, w_{b-2} \}$.  Node 1 sends the superposition/sum of  ${\bf X_1'} (w_b)$, ${\bf X_2'}(w_{b-1})$ and ${\bf X_3'}(w_{b-2})$ with power $\alpha_1 P_1$, $\beta_1 P_1$, and $(1 - \alpha_1 - \beta_1) P_1$ respectively. Node 2 sends the superposition/sum of  $\sqrt{\frac{\alpha_2P_2}{\beta_1P_1}} {\bf X_2'} (w_{b-1})$ and $\sqrt{\frac{(1-\alpha_2) P_2}{(1-\alpha_1 - \beta_1)P_1}} {\bf X_3' }(w_{b-2})$ with power $\alpha_2 P_2$, and $(1-\alpha_2) P_2$ respectively. Node 3 sends $\sqrt{\frac{P_3}{(1 - \alpha_1 - \beta_1) P_1}} {\bf X_3'} (w_{b-2})$ with power $P_3$.

\smallskip
\noindent{\bf Decoding:}

 {\it Node 2 decodes $w_b$:} In block $b$, since Node 2 knows $w_{b-1}$ and $w_{b-2}$ and thus ${\bf X_2'} (w_{b-1})$ and ${\bf X_3'} (w_{b-2})$, it can subtract these terms from its received signal 
\[ {\bf Y_2}(b) = {\bf X_1'}(w_b) + {\bf X_2'}(w_{b-1}) + {\bf X_3'} (w_{b-2}) +  \sqrt{\frac{P_3}{(1 - \alpha_1 - \beta_1) P_1}} { \bf X_3'}(w_{b-2}) + {\bf Z_2}(b)\]
 and obtains a noisy observation of ${\bf X'_1}(w_b)$ only. 
Node 2 is able to then uniquely decode $w_b$ as long as (see \cite{Erez:2004} or Lemma \ref{lem:unique})
\[ R < \frac{1}{2} \log \left( 1 + \frac{\alpha_1 P_1}{N_2} \right). \]

{\it Node 3 decodes $w_{b-1}$:} Since Node 3 knows $w_{b-2}$ and thus ${\bf X_3'}(w_{b-2})$, it subtracts these from ${\bf Y_3}(b)$:
\[ {\bf Y_3}(b) = {\bf X_1'} (w_b) + {\bf X_2'} (w_{b-1}) + {\bf X_3'}(w_{b-2}) + \sqrt{\frac{\alpha_2P_2}{\beta_1P_1}} {\bf X_2'}(w_{b-1}) + \sqrt{\frac{(1-\alpha_2) P_2}{(1-\alpha_1 - \beta_1)P_1}}{\bf X_3'} (w_{b-2}) + {\bf Z_3}(b)\]  
and obtains a noisy observation of ${\bf X_1'} (w_{b})$ and ${\bf X_2'}(w_{b-1})$, 
\[ {\bf Y_3^*}(b) =  {\bf X_1'} (w_b) +  \left(1+\sqrt{\frac{\alpha_2P_2}{\beta_1P_1}}\right) {\bf X_2'}(w_{b-1}) + {\bf Z_3}(b). \] It then uses $\Lambda_{s(2-3)}$ to decode a list $L_{2-3}^{w_{b-1}}({\bf Y_3^*(b)})$ of possible $w_{b-1}$  of size $2^{n\left(R- C\left(\frac{ \left(\sqrt{\beta_1P_1} + \sqrt{\alpha_2P_2}\right)^2}{\alpha_1P_1 + N_3}\right) \right)}$ in the presence of interference ${\bf X_1'} (w_{b})$ (uniformly distributed over the fundamental Voronoi  region of a Rogers good lattice code) and Gaussian noise ${\bf Z_3}(b)$ (hence we may apply Theorem \ref{thm:list}). 
It then intersects this list $L_{2-3}^{w_{b-1}}({\bf Y_3^*}(b))$  with the list $L_{1-3}^{w_{b-1}}({\bf Y_3^{**}}(b-1))$  of asymptotic size  $2^{n \left( R - C\left(\frac{\alpha_1P_1}{N_3} \right) \right) }$ obtained in block $b-1$ by subtracting off the known signals dependent on $w_{b-2}, w_{b-3}$ to obtain ${\bf Y_3^{**}}(b-1) = {\bf X_1'}(w_{b-1}) + {\bf Z_3}(b-1)$. 
To ensure a unique $w_{b-1}$ in the intersection, by independence of the lists (based on the independent mappings of the messages to the codebooks  ${\cal C}_1$ and ${\cal C}_2$), we need
\begin{align*}
R &<  C\left(\frac{ \left(\sqrt{\beta_1P_1} + \sqrt{\alpha_2P_2}\right)^2}{\alpha_1P_1 + N_3}\right)+ C\left(\frac{\alpha_1P_1}{N_3} \right) \\
&= C\left(\frac{\alpha_1P_1 + \left(\sqrt{\beta_1P_1} + \sqrt{\alpha_2P_2}\right)^2}{ N_3}\right).
\end{align*}
After Node 3 decodes $w_{b-1}$, it further subtracts ${\bf X_2'} (w_{b-1})$ from its received signal and obtains a noisy observation of ${\bf X_1'} (w_{b})$. It again uses the lattice list decoder using $\Lambda_{s(1-3)}$ to output a list $L_{1-3}^{w_{b}}({\bf Y_3^{**}}(b))$ of $w_{b}$ of size  $2^{n \left( R - C\left(\frac{\alpha_1P_1}{N_3} \right) \right) }$ which is used in block $b+1$ to determine $w_b$.

{\it Node 4 decodes $w_{b-2}$:} Finally, Node 4 intersects three lists to determine $w_{b-2}$. These three lists are again independent by the independent mapping of the messages to the codebooks ${\cal C}_1$, ${\cal C}_2$, ${\cal C}_3$,  where each corresponds to one of the three links (between node 1-4, 2-4, and 3-4). The first list $L_{3-4}^{w_{b-2}}({\bf Y_4}(b))$ of $w_{b-2}$ messages is obtained by list decoding using $\Lambda_{s(3-4)}$ on  its received signal 
\begin{align*} {\bf Y_4}(b) =& {\bf X_1'} (w_b) + {\bf X_2'} (w_{b-1}) + {\bf X_3'} (w_{b-2}) + \sqrt{\frac{\alpha_2P_2}{\beta_1P_1}} {\bf X_2'} (w_{b-1}) \\ &+ \sqrt{\frac{(1-\alpha_2) P_2}{(1-\alpha_1 - \beta_1)P_1}} {\bf X_3'} (w_{b-2}) + \sqrt{\frac{P_3}{(1 - \alpha_1 - \beta_1) P_1}} {\bf X_3'}(w_{b-2}) + {\bf Z_4}(b) \end{align*}
which is a combination of scaled signals 
${\bf X_1'}(w_b)$ and $ {\bf X_2'} (w_{b-1})$ which are uniform over the fundamental Voronoi  regions of Rogers good lattices and additive Gaussian noise ${\bf Z_4}(b)$, and is of size 
 \[ |L_{3-4}^{w_{b-2}}({\bf Y_4}(b))|  = 2^{n\left(R- C\left(\frac{ \left(\sqrt{(1 - \alpha_1-\beta_1P_1)} + \sqrt{(1-\alpha_2)P_2} +\sqrt{P_3} \right)^2}{\alpha_1P_1 + \left(\sqrt{\beta_1P_1} + \sqrt{\alpha_2P_2}\right)^2 + N_4}\right) \right)}.\] 
 The second list  $L_{2-4}^{w_{b-2}}({\bf Y_4^*}(b-1))$ is obtained in block $b-1$ and is of size $2^{n\left(R- C\left(\frac{ \left(\sqrt{\beta_1P_1} + \sqrt{\alpha_2P_2}\right)^2}{\alpha_1P_1 + N_4}\right) \right)}$, while the third list $L_{1-4}^{w_{b-2}}({\bf Y_4^{**}}(b-2))$ is obtained in block $b-2$ and is of size  $2^{n \left( R - C\left(\frac{\alpha_1P_1}{N_4} \right) \right)}$. The formation of these lists is described next (they are formed analogously in blocks $b-1$ and $b-2$). 

After the successful decoding of $w_{b-2}$ in block $b$,  node 4 decodes two more lists which are used in the blocks $b+1$ and $b+2$ to determine $w_{b-1}$ and $w_b$ respectively. 
Node 4 first subtracts the ${\bf X_3'}(w_{b-2})$ terms from its received signal ${\bf Y_4}(b)$ to obtain ${\bf Y_4^*}(b)$ and decodes a list of possible $w_{b-1}$ from the terms $ {\bf X_2'}(w_{b-1})$ using $\Lambda_{s(2-4)}$ in the presence of interference terms ${\bf X_1'} (w_b)$ which are uniformly distributed over Rogers good lattices and Gaussian noise (hence Theorem \ref{thm:list} applies). 
This list is denoted as $L_{2-4}^{w_{b-1}}({\bf Y_4^*}(b))$ and is used in the block $b+1$ to determine $w_{b-1}$. 

After Node 4 decodes $w_{b-1}$ in the block $b+1$, it further subtracts the ${\bf X_2'} (w_{b-1})$ terms from ${\bf Y_4^*}(b)$  to obtain ${\bf Y_4^{**}}(b) = {\bf X_1'}(w_b) + {\bf Z_4}(b)$. It then uses $\Lambda_{s(1-4)}$ to decode a list of $w_b$,   denoted as $L_{1-4}^{w_b}({\bf Y_4^{**}}(b))$, which is used in block $b+2$ to determine $w_b$.  

In block $b$, to ensure a unique message $w_{b-2}$ in the intersection of the three independent lists, we need
\begin{align*}
R &< C\left(\frac{ \left(\sqrt{(1 - \alpha_1-\beta_1P_1)} + \sqrt{(1-\alpha_2)P_2} +\sqrt{P_3} \right)^2}{\alpha_1P_1 + \left(\sqrt{\beta_1P_1} + \sqrt{\alpha_2P_2}\right)^2 + N_4}\right) + C\left(\frac{ \left(\sqrt{\beta_1P_1} + \sqrt{\alpha_2P_2}\right)^2}{\alpha_1P_1 + N_4}\right) + C\left(\frac{\alpha_1P_1}{N_4} \right) \\
&= C\left(\frac{\alpha_1P_1 + \left(\sqrt{\beta_1P_1} + \sqrt{\alpha_2P_2}\right)^2 + \left(\sqrt{(1 - \alpha_1-\beta_1P_1)} + \sqrt{(1-\alpha_2)P_2} +\sqrt{P_3} \right)^2 }{ N_4}\right).
\end{align*}
\end{proof}

\end{appendix}

\bibliographystyle{IEEEtran}
\bibliography{refs}
\end{document}